\documentclass[10pt]{article} 

\usepackage{amsfonts,amssymb,amsmath,stmaryrd,latexsym}
\usepackage{enumerate,array,xspace}

\usepackage[hidelinks]{hyperref}
\usepackage{graphicx}
\usepackage{multirow}

\newenvironment{tbs}{%
   \small\tt
   \begin{itemize}}{\end{itemize}}
\newcommand{\btbs}{\begin{tbs}}                                                                      
\newcommand{\etbs}{\end{tbs}}

\newcommand{\hide}[1]{}

\newtheorem{theorem}{Theorem}[section]

\newtheorem{fact}[theorem]{Fact}
\newtheorem{proposition}[theorem]{Proposition}

\newtheorem{corollary}[theorem]{Corollary}
\newtheorem{defi}[theorem]{Definition}

\newtheorem{conv}[theorem]{Convention}
\newtheorem{rema}[theorem]{Remark}
\newtheorem{exam}[theorem]{Example}

\newenvironment{definition}{\begin{defi}\rm}{\hfill $\lhd$\end{defi}}
\newenvironment{convention}{\begin{conv}\rm}{\end{conv}}
\newenvironment{remark}{\begin{rema}\rm}{\hfill $\lhd$\end{rema}}

\newenvironment{proof}{\begin{trivlist}\item[]{\bf
Proof.}}{\hfill {\sc qed}\end{trivlist}}

\newenvironment{proofof}[1]{\begin{trivlist}\item[\hskip\labelsep{\bf
Proof~of~{#1}.\ }]}{\hspace*{\fill} {\sc qed}\end{trivlist}}

\newtheorem{claim2}{\sc Claim}
\newenvironment{claim}{\begin{claim2}\rm}{\end{claim2}\rm}
\newenvironment{claimfirst}{\setcounter{claim2}{0}
               \begin{claim2}\rm}{\end{claim2}\rm}

\newenvironment{pfclaim}{\begin{trivlist}\item[]{\sc Proof of
Claim}}{\hfill {\mbox{$\blacktriangleleft$}}\end{trivlist}}

\newcommand{\Dom}{\mathsf{Dom}}

\newcommand{\pto}{\rightharpoonup}

\newcommand{\nada}{\varnothing}

\newcommand{\sz}[1]{|#1|}

\newcommand{\vlist}[1]{\overline{{\mathbf{#1}}}}

\newcommand{\isdef}{\mathrel{:=}}

\newcommand{\osmodel}{\mathbb{D}}
\newcommand{\emodel}{{\osmodel_\nada}}
\newcommand{\rest}{\upharpoonright}

\newcommand{\umods}{\mathfrak{M}}

\newcommand{\fovar}{\mathsf{iVar}}
\newcommand{\defbnf}{\mathrel{::=}}

\newcommand{\foeq}{\approx}
\newcommand{\foneq}{\not\approx}
\newcommand{\qu}{\exists^\infty}
\newcommand{\dqu}{\forall^\infty}
\newcommand{\wqu}{\ensuremath{\mathbf{W}}\xspace}

\newcommand{\dbnf}{\nabla}
\newcommand{\dbnfofo}[1]{\dbnf_{\ofo}(#1)}
\newcommand{\dgbnfofo}[2]{\dbnf_{\ofo}(#1,#2)}
\newcommand{\dbnfofoe}[2]{\dbnf_{\ofoe}(#1,#2)}
\newcommand{\dbnfofoei}[3]{\dbnf_{\ofoei}(#1,#2,#3)}
\newcommand{\dbnfinf}[1]{\dbnf_{\!\!\infty}(#1)}

\newcommand{\mondbnfofo}[2]{\dbnf^{#2}_\ofo(#1)}
\newcommand{\mondgbnfofo}[3]{\dbnf^{#3}_\ofo(#1,#2)}
\newcommand{\mondbnfofoe}[3]{\dbnf^{#3}_{\ofoe}(#1,#2)}
\newcommand{\mondbnfofoei}[4]{\dbnf^{#4}_{\ofoei}(#1,#2,#3)}
\newcommand{\mondbnfinf}[2]{\dbnf^{#2}_\infty(#1)}

\newcommand{\posdbnfofo}[1]{\dbnf^+_{\ofo}(#1)}
\newcommand{\posdgbnfofo}[2]{\dbnf^+_{\ofo}(#1,#2)}
\newcommand{\posdbnfofoe}[2]{\dbnf^+_{\ofoe}(#1,#2)}
\newcommand{\posdbnfofoei}[3]{\dbnf^+_{\ofoei}(#1,#2,#3)}

\newcommand{\arediff}[1]{\mathrm{diff}(#1)}
\newcommand{\hs}{\heartsuit}

\newcommand{\llang}{\mathtt{L}} 

\newcommand{\ofo}{\mathtt{M}}
\newcommand{\ofoe}{\ofo\mathtt{E}}
\newcommand{\ofoei}{\ofoe^{\infty}}

\newcommand{\monot}[2]{\mathtt{Pos}_{#2}(#1)}
\newcommand{\cont}[2]{\mathtt{Con}_{#2}(#1)}

\newcommand{\univ}[1]{\mathtt{Univ}(#1)}
\newcommand{\aut}[1]{\mathtt{Aut}(#1)}

\newcommand{\tbas}[1]{#1^{*}}
\newcommand{\tmono}{\oslash}
\newcommand{\tcont}{\ominus}

\newcommand{\tmoda}{\circ}
\newcommand{\tmodb}{\bullet}
\newcommand{\tuniv}{\otimes}
\newcommand{\tinvq}{\oast}

\newcommand{\qr}{\mathtt{qr}}

\newcommand{\ext}[1]{\llbracket#1\rrbracket}

\newcommand{\efgame}{\mathrm{EF}}

\newcommand{\eloise}{\ensuremath{\exists}\xspace}
\newcommand{\abelard}{\ensuremath{\forall}\xspace}

\newcommand{\bbN}{\mathbb{N}}

\renewcommand{\phi}{\varphi}

\title{Model Theory of Monadic Predicate Logic with the Infinity Quantifier
}

\author{
Facundo Carreiro
\thanks{Institute for Logic, Language and Computation, Universiteit van Amsterdam,
   P.O. Box 94242, 1090 GE Amsterdam. E-mail: \url{contact@facundo.io}.}
\and Alessandro Facchini\thanks{Dalle Molle Institute for Artificial Intelligence (IDSIA),
Galleria 2, 6928 Manno (Lugano), Switzerland. E-mail: \url{alessandro.facchini@idsia.ch}.}
\and Yde Venema
\thanks{Institute for Logic, Language and Computation, Universiteit van Amsterdam,
   P.O. Box 94242, 1090 GE Amsterdam. E-mail: \url{y.venema@uva.nl}.}
\and Fabio Zanasi
\thanks{University College London,
66-72 Gower Street, WC1E 6BT London, United Kingdom. E-mail: \url{f.zanasi@ucl.ac.uk}.}
}

\date{\today}

\begin{document}

\maketitle



\begin{abstract}

This paper establishes model-theoretic properties of $\ofoei$, a variation of monadic first-order logic that features the generalised quantifier $\qu$ (`there are infinitely many').

We provide syntactically defined fragments of $\ofoei$ characterising four different semantic properties of $\ofoei$-sentences: (1) being monotone and (2) (Scott) continuous in a given set of monadic predicates; (3) having truth preserved under taking submodels or (4) invariant under taking quotients. In each case, we produce an effectively defined map that 
translates an arbitrary sentence $\varphi$ to a sentence $\varphi^{\sf p}$ belonging to the
corresponding syntactic fragment, with the property that $\varphi$ is equivalent to $\varphi^{\sf p}$ precisely when it has the associated semantic property. 

Our methodology is first to provide these results in the simpler setting of monadic first-order logic with ($\ofoe$) and without ($\ofo$) equality, and then move to $\ofoei$ by including the generalised quantifier $\qu$ into the picture.

As a corollary of our developments, we obtain that the four semantic properties above are decidable for $\ofoei$-sentences. Moreover, our results are directly relevant to the characterisation of automata and expressiveness modulo bisimilirity for variants of monadic second-order logic. This application is developed in a companion paper.
\end{abstract}



\section{Introduction}\label{sec:intro}

Model theory investigates the relationship between formal languages and 
semantics. 
From this perspective, among the most important results are the so called 
\emph{preservation theorems}. 
Such results typically characterise a certain language as the fragment of 
another, richer language satisfying a certain model-theoretic property.
In doing so, they therefore  link the syntactic shape of a formula with the 
semantic properties of the class of models  it defines. 
In the case of classical first-order logic, notable examples are the 
{\L}o\'s-Tarski theorem,  stating that a first-order  formula is equivalent to a universal one if and only if the class of its models is closed under taking submodels, and Lyndon's theorem, stating that a first-order formula is equivalent to  one for which each occurrence of a relation symbol $R$  is positive if and only if it is monotone with respect to the interpretation of $R$ (see e.g. \cite{Hodges1993}).

The aim of this paper is to show that similar results also hold when considering the predicate logic $\ofoei$ that allows only monadic predicate symbols and no function symbols, but that goes beyond standard first-order logic with equality in that it features the generalised quantifier `there are infinitely many'. 
 
Generalised quantifiers were introduced by Mostowski in \cite{Mostowski1957}, 
and in a more general sense by Lindstr\"{o}m in \cite{perlindstrom1966first},
the main motivation being the observation that standard first-order quantifiers 
`there are some' and `for all' are not sufficient for expressing some basic 
mathematical concepts. 
Since then, they have attracted a lot of interests, insomuch that their study
constitutes nowadays a well-established field of logic with important 
ramifications in disciplines such as linguistics and computer 
science.\footnote{%
   For an overview see e.g. 
   \cite{van1995directions,vaananen1997generalized,sep-generalized-quantifiers}.
   For an introduction to the model theory of generalised quantifiers, the 
   interested reader can consult for 
   instance~\cite[Chapter~10]{vaananen2011models}.
   }.

Despite the fact that the absence of polyadic predicates clearly restricts its expressing power, monadic first-order logic (with identity) displays nice properties, both from a computational and a model-theoretic point of view. Indeed, the  
satisfiability problem becomes decidable  \cite{Behmann1922,Loewenheim1915}, and, in addition of an immediate application of {\L}o\'s-Tarski  and Lyndon's theorems, one can also obtain a Lindstr\"om like characterisation result  \cite{tharp1973characterization}.
Moreover, adding the possibility of quantifying over predicates does not increase the expressiveness of the language \cite{ackermann1954solvable}, meaning that when restricted to monadic predicates, monadic second order logic collapses into first-order logic. 

For what concerns monadic first-order logic extended with an infinity quantifier,
in \cite{Mostowski1957} Mostowski, already proved its decidability, whereas from 
work of V\"a\"an\"anen  \cite{vaananen77} we know that its expressive power 
coincides with that of weak monadic second-order logic restricted to monadic 
predicates, that is monadic first-order logic extended with a second order 
quantifier ranging over finite sets\footnote{%
   Extensions of monadic first-order logic with other generalised quantifiers
   have also been studied (see 
   e.g.~\cite{slomson1968monadic,caicedo1981extensions}).
   }.

\subsection*{Preservation results and proof outline.}

A preservation result involves some fragment $\llang_{\mathfrak{P}}$ of a given
yardstick logic $\llang$, related to a certain semantic property $\mathfrak{P}$. 
It is usually formulated as
\begin{equation}\label{eq:intro-0}
\phi  \in \llang \text{ has the property } \mathfrak{P} \text{ iff } 
\phi \text{ is equivalent to some } \phi' \in \llang_{\mathfrak{P}}.
\end{equation}

In this work, our main yardstick logic will be $\ofoei$. 
Table \ref{tab:0} summarises the semantic properties
($\mathfrak{P}$) we are going to consider,  
the corresponding expressively complete fragment ($\llang_{\mathfrak{P}}$) and preservation theorem.

\begin{table}[h!]
\begin{center}
\begin{tabular}{|c|c|c|}
\hline
$\mathfrak{P}$				
   & $\llang_{\mathfrak{P}}$ 		
   & Preservation Theorem 
\\ \hline \hline
     Monotonicity     
   & Positive fragment         
   & Theorem \ref{t:mono}                        
\\ (Definition \ref{def:mono})		
   & $\monot{\ofoei}{}$
   &
\\ \hline	
     Continuity     
   & Continuous fragment  
   & Theorem \ref{t:cont}            
\\ (Definition \ref{def:cont})       	
   & $\cont{\ofoei}{}$        
   &                              
\\\hline
     Preservation under submodels
   & Universal fragment    	 
   &          Theorem \ref{t:univ}                    
\\ (Definition \ref{d:inv}(\ref{d:inv}))    
   &          $\univ{\ofoei}{}$          
   &                              
\\ \hline
     Invariance under quotients& Monadic first-order logic
   & Theorem \ref{t:qinv}                 
\\ (Definition \ref{d:inv}(\ref{d:inv}))   	
   & $\ofo$             	
   &                              
\\\hline
\end{tabular}
\caption{A summary of our preservation theorems}
\label{tab:0}
\end{center}
\end{table}

The proof of each  preservation theorem is composed of two parts. 
The first, simpler one concerns 
the claim that each sentence in the fragment satisfies the concerned property. 
It is usually proved  by induction on the structure of the sentence.  
The other direction is the  \emph{expressive completeness statement}, stating 
that within the considered logic, the fragment is expressively complete for the 
property. 
Its verification generally requires more effort. 
In this paper, we will actually verify a stronger expressive completeness 
statement. 
Namely, for each semantic property $\mathfrak{P}$ and corresponding fragment
$\llang_{\mathfrak{P}}$ from Table \ref{tab:0}, we are going to provide an 
effective translation operation $(\cdot)^{\sf p}: \ofoei \to 
\llang_{\mathfrak{P}}$ such that \begin{equation}\label{eq:intro-i}
\text{if }\phi \in \ofoei\text{ has the property } \mathfrak{P} \text{ then }
\phi \text{ is equivalent to } \phi^{\sf p}.
\end{equation}
Since the satisfiability problem for $\ofoei$ is decidable and the translation
$(\cdot)^{\sf p}$ is effectively computable, we obtain, as an immediate 
corollary of (\ref{eq:intro-i}), that for each property $\mathfrak{P}$ listed
in Table \ref{tab:0}
\begin{equation}\label{eq:intro-ii}
\text{the problem whether a $\ofoei$-sentence satisfies property $\mathfrak{P}$
or not is decidable.}
\end{equation}

The proof of each instance of (\ref{eq:intro-i})  will follow an uniform 
pattern, analogous to the one employed in the aim of obtaining similar results
in the context of the modal $\mu$-calculus \cite{Jan96,d2000logical,FV12}.
The crux of the adopted proof method is that, extending known results on monadic
first-order logic, for each sentence $\phi$ in $ \ofoei$ it is possible to compute
a logically equivalent sentence in \emph{basic normal norm}. 
Such normal forms will take the shape of a disjunction $\bigvee \dbnf_{\ofoei}$, 
where each disjunct $\dbnf_{\ofoei}$ characterises a class of models of $\phi$
satisfying the same set of ${\ofoei}$-sentences of equal quantifier rank as 
$\phi$. 
Based on this, it will therefore be enough to define an effective translation
$(\cdot)^{\sf p}$ for sentences in normal form, point-wise in each disjunct 
$\dbnf_{\ofoei}$, and  then verify that it indeed satisfies (\ref{eq:intro-i}). 

As a corollary of the employed proof method, we thus obtain effective normal
forms for sentences satisfying the considered property.

In addition to $\ofoei$, in this paper we are also going to consider monadic 
first-order logic with and without equality, denoted respectively by $\ofoe$ 
and $\ofo$.
Table \ref{tab:1} shows a summary of the expressive completeness and normal 
form results presented in this paper. 

\begin{table}[h!]
\begin{center}
\begin{tabular}{cc|c|c|c|c|}
\cline{3-5}
& & \multicolumn{3}{ c| }{Language} \\ \cline{3-5}
& 									& $\ofo$ 					& $\ofoe$ 				& $\ofoei$  			\\ \cline{2-5}
& \multicolumn{1}{ |c| }{ Normal forms}		& Fact \ref{fact:ofonormalform} &  Thm. \ref{thm:bnfofoe} 	& Thm. \ref{thm:bfofoei} \\ \hline
\multicolumn{1}{ |c  }{\multirow{2}{*}{Monotonicity} }&
\multicolumn{1}{ |c| }{Completeness}   & Prop. \ref{p:fomon} & Prop. \ref{p:monofoeismonot} & Prop. \ref{p:mono-ofoei} \\ \cline{2-5}		
\multicolumn{1}{ |c  }{}                        &
\multicolumn{1}{ |c| }{Normal forms} 	& Cor. \ref{cor:ofopositivenf}	 		& Cor. \ref{cor:ofoepositivenf}		& Cor. \ref{cor:ofoeipositivenf}  \\ \cline{1-5}		
\multicolumn{1}{ |c  }{\multirow{2}{*}{Continuity} }&
\multicolumn{1}{ |c| }{Completeness}   & Prop. \ref{prop:ofocont}  & Fact \ref{fact:vb}  & Prop. \ref{lem:ofoeictrans} \\ \cline{2-5}		
\multicolumn{1}{ |c  }{}                        &
\multicolumn{1}{ |c| }{Normal forms} 	& Cor. \ref{cor:ofocontinuousnf}	 		& --		& Cor. \ref{cor:ofoeicontinuousnf}  \\ \cline{1-5}			
\multicolumn{1}{ |c  }{\multirow{2}{*}{Preservation under submodels} }&
\multicolumn{1}{ |c| }{Completeness}   & \multicolumn{3}{ c| }{Prop. \ref{p:univ2} }  \\ \cline{2-5}		
\multicolumn{1}{ |c  }{}                        &
\multicolumn{1}{ |c| }{Normal forms} 	& Cor. \ref{cor:univ}(\ref{cor:ofo})	 		& Cor. \ref{cor:univ}(\ref{cor:ofoe})		& Cor. \ref{cor:univ}(\ref{cor:ofoei})  \\ \cline{1-5}	
\multicolumn{1}{ |c  }{\multirow{2}{*}{Invariance under quotients } }&
\multicolumn{1}{ |c| }{Completeness}   & Prop. \ref{p:m-qinv} &\multicolumn{2}{ c| }{Prop. \ref{p:invq} }  \\ \cline{2-5}		
\multicolumn{1}{ |c  }{}                        &
\multicolumn{1}{ |c| }{Normal forms} 	& Fact \ref{fact:ofonormalform} & \multicolumn{2}{ c| }{Cor. \ref{cor:qinv} }  \\ \cline{1-5}	
\end{tabular}
\caption{An overview of our expressive completeness and normal form results.}
\label{tab:1}
\end{center}
\end{table}

\subsection*{Application of obtained results: the companion paper}
\emph{Parity automata} are finite-state systems playing a crucial role in 
obtaining decidability and expressiveness results in fixpoint logic (see e.g. 
\cite{vardi2008automata}).
They are specified by a finite set of states $A$, a distinguished, initial state
$a \in A$, a function $\Omega$ assigning to each states a priority (a natural 
number), and a transition function $\Delta$ whose co-domain is usually given by
a monadic logic in which  the set of (monadic) predicates  coincides with $A$.
Hence, each monadic logic $\llang$ induces its own class of automata 
$\aut{\llang}$.

A landmark result in this area is Janin and Walukiewicz's theorem stating that
the bisimulation-invariant fragment of monadic second order logic coincides 
with the modal $\mu$-calculus \cite{Jan96}, and the proof of this result is an 
interesting mix of the theory of parity automata and the model theory of
monadic predicate logic.
First, preservation and normal forms results are used to verify that (on tree 
models) $\aut{\monot{\ofoe}{}}$ is the class of automata characterising the
expressive power of monadic second order logic \cite{Walukiewicz96}, whereas 
$\aut{\monot{\ofo}{}}$ corresponds to the modal $\mu$-calculus \cite{JaninW95},
where $\monot{\llang}{}$ denote the positive fragment of the monadic logic 
$\llang$. 
Then, Janin-Walukieiwcz's expressiveness theorem is a consequence of these 
automata characterisations and the fact that positive monadic first-order 
logic without equality provides the quotient-invariant fragment of positive
monadic first-order logic with equality (see Theorem \ref{t:inv1}).

In our companion paper \cite{companionWEAK}, we provide a Janin-Walukiewicz 
type characterisation result for \emph{weak} monadic second order logic.
Our proof, analogously to the case of full monadic second order logic discussed 
previously,  crucially employs preservation and normal form results for $\ofoei$
listed in  Tables \ref{tab:0} and \ref{tab:1}.

\subsection*{Other versions}
Results in this paper first appeared in the first author's PhD thesis
(\cite[Chapter 5]{carreiro2015fragments}); this journal version largely expands
material first published as part of the conference 
papers \cite{DBLP:conf/lics/FacchiniVZ13,carreiro2014weak}. 
In particular, the whole of Section \ref{sec:inv} below contains new results.



\section{Basics}

In this section we provide the basic definitions of the monadic predicate 
liftings that we study in this paper.
Throughout this paper we fix a finite set $A$ of objects that we shall refer to 
as \emph{(monadic) predicate symbols} or \emph{names}.

We shall also assume an infinite set $\fovar$ of \emph{individual variables}.

\begin{definition}
Given a finite set $A$ we define a \emph{(monadic) model} to be a pair
$\osmodel = (D,V )$ consisting of a set $D$,  which we call the \emph{domain} of 
$\osmodel$, and an interpretation or \emph{valuation} $V : A \to \wp (D)$.
The class of all models will be denoted by $\umods$.
\end{definition}

\begin{remark}~
Note that we make the somewhat unusual choice of allowing the domain of a
monadic model to be \emph{empty}.
In view of the applications of our results to automata theory (see 
Section \ref{sec:intro}) this choice is very natural, even if it means that 
some of our proofs here become more laborious in requiring an extra check.
Observe that there is exactly one monadic model based on the empty domain;
we shall denote this model as $\emodel \isdef (\nada, \nada)$. 
\end{remark}

\begin{definition}
Observe that a valuation $V: A \to \wp (D)$ can equivalently be presented via 
its associated \emph{colouring} $V^{\flat}:D \to \wp(A)$ given by
\[
V^{\flat}(d) \isdef \{a \in A \mid d \in V(a)\}.
\]
We will use these perspectives interchangeably, calling the set $V^{\flat}(d)
\subseteq A$ the \emph{colour} or \emph{type} of $d$.
In case $D = \nada$, $V^{\flat}$ is simply the empty map.
\end{definition}

In this paper we study three languages of monadic predicate logic: the languages 
$\ofoe$ and $\ofo$ of monadic first-order logic with and without equality, 
respectively, and the extension $\ofoei$ of $\ofoe$ with the 
generalised quantifiers $\qu$ and $\dqu$.
Probably the most concise definition of the full language of monadic predicate 
logic would be given by the following grammar:
\[
\varphi \defbnf a(x)
\mid x \foeq y
\mid \neg \varphi
\mid \varphi \lor \varphi
\mid \exists x.\varphi
\mid \qu x.\varphi,
\]
where $a \in A$ and $x$ and $y$ belong to the set $\fovar$ of individual 
variables.
In this set-up we would need to introduce the quantifiers $\forall$ and $\dqu$ 
as abbreviations of $\neg \exists \neg$ and $\neg \qu \neg$, respectively.
However, for our purposes it will be more convenient to work with a variant of 
this language where all formulas are in negation normal form; that is, we only 
permit the occurrence of the negation symbol $\neg$ in front of an atomic 
formula.
In addition, for technical reasons we will add $\bot$ and $\top$ as constants,
and we will write $\neg (x \foeq y)$ as $x \foneq y$.
Thus we arrive at the following definition of our syntax.

\begin{definition}
The set $\ofoei(A)$ of \emph{monadic formulas} is given by the following grammar:
\[
\varphi \defbnf
\top \mid \bot 
\mid a(x)
\mid \neg a(x)
\mid x \foeq y
\mid x \foneq y
\mid \varphi \lor \varphi
\mid \varphi \land \varphi
\mid \exists x.\varphi
\mid \forall x.\varphi
\mid \qu x.\varphi
\mid \dqu x.\varphi
\]
where $a \in A$ and $x,y\in \fovar$. 
The language $\ofoe(A)$ of \emph{first-order logic with equality} is defined as
the fragment of $\ofoei(A)$ where occurrences of the generalised quantifiers 
$\qu$ and $\dqu$ are not allowed:
\[
\varphi \defbnf
\top \mid \bot 
\mid a(x)
\mid \neg a(x)
\mid x \foeq y
\mid x \foneq y
\mid \varphi \lor \varphi
\mid \varphi \land \varphi
\mid \exists x.\varphi
\mid \forall x.\varphi
\]
Finally, the language $\ofo(A)$ of \emph{first-order logic} is the equality-free
fragment of $\ofoe(A)$; that is, atomic formulas of the form $x \foeq y$ and 
$x \foneq y$ are not permitted either:
\[
\varphi \defbnf
\top \mid \bot 
\mid a(x)
\mid \neg a(x)
\mid \varphi \lor \varphi
\mid \varphi \land \varphi
\mid \exists x.\varphi
\mid \forall x.\varphi
\]

In all three languages we use the standard definition of free and bound
variables, and we call a formula a \emph{sentence} if it has no free variables.
In the sequel we will often use the symbol $\llang$ to denote either of the 
languages $\ofo$, $\ofoe$ or $\ofoei$.

For each of the languages $\llang \in \{ \ofo, \ofoe, \ofoei \}$, we define the 
\emph{positive fragment} $\llang^{+}$ of $\llang$ as the language obtained by 
almost the same grammar as for $\llang$, but with the difference that we do not
allow negative formulas of the form $\neg a(x)$.
\end{definition}

\noindent

The semantics of these languages is given as follows. 

\begin{definition}
The semantics of the languages $\ofo, \ofoe$ and $\ofoei$ is given in the form
of a truth relation $\models$ between models and sentences of the language.
To define the truth relation on a model $\osmodel = (D ,V)$, we distinguish 
cases.

\begin{description}
\item[Case $D=\nada$:] 
We define the \emph{truth relation} $\models$ on the empty model $\emodel$ for
all formulas that are Boolean combinations of sentences of the form $Qx. \phi$,
where $Q \in \{ \exists, \qu, \forall, \dqu \}$ is a quantifier.
The definition is by induction on the complexity of such sentences; the 
``atomic'' clauses, where the sentence is of the form $Qx. \phi$, is as follows:
\[\begin{array}{lll}
 \emodel \not\models Qx. \phi 
   & \text{if}\quad Q \in \{ \exists,  \qu \},&
\\   
\emodel \models Qx. \phi 
   & \text{if}\quad Q \in \{ \forall,  \dqu \}.&
\end{array}\]
The clauses for the Boolean connectives are as expected.

\item[Case $D \neq \nada$:] 
In the (standard) case of a non-empty model $\osmodel$, we extend the truth 
relation to arbitrary formulas, involving assignments of individual variables 
to elements of the domain.
That is, given a model $\osmodel = (D,V)$, an assignment $g :\fovar\to D$ and 
a formula $\phi \in \ofoei(A)$ we define the \emph{truth} relation $\models$ 
by a straightforward induction on the complexity of $\phi$.
Below we explicitly provide the clauses of the quantifiers:
\[\begin{array}{lll}
   \osmodel,g \models \exists x.\varphi 
   & \text{iff}\quad \osmodel,g [x\mapsto d] \models \varphi 
     \text{ for some $d\in D$},
\\   \osmodel,g \models \forall x.\varphi 
   & \text{iff}\quad \osmodel,g [x\mapsto d] \models \varphi 
     \text{ for all $d\in D$},
\\   \osmodel,g \models \qu x.\varphi 
   & \text{iff}\quad \osmodel,g [x\mapsto d] \models \varphi 
     \text{ for infinitely many $d\in D$},
\\   \osmodel,g \models \dqu x.\varphi 
   & \text{iff}\quad \osmodel,g [x\mapsto d] \models \varphi 
     \text{ for all but at most finitely many $d\in D$}.
\end{array}\]
The clauses for the atomic formulas and for the Boolean connectives are standard.
\end{description}
In what follows, when discussing the truth of $\phi$ on the empty model, we
always implicitly assume that $\phi$ is a sentence.
\end{definition}

As mentioned in the introduction, general quantifiers such as $\qu$ and $\dqu$
were introduced by Mostowski~\cite{Mostowski1957}, who proved the decidability 
for the language obtained by  extending $\ofo$ with such quantifiers. 
The decidability of the full language $\ofoei$ was then proved by Slomson
in \cite{slomson1968monadic}.\footnote{%
   The argument in \cite{slomson1968monadic} is given in terms of the so called
   \emph{Chang quantifier} but is easily seen to work also for $\qu$ and $\dqu$.
   Both Mostowski's and Slomson's decidability results can be extended to the 
   case of the empty domain.}
The case for $\ofo$ and $\ofoe$ goes back already to
\cite{Behmann1922,Loewenheim1915}.

\begin{fact}\label{f:decido}
For each logic $\llang \in \{ \ofo, \ofoe, \ofoei \}$, the problem of whether a
given $\llang$-sentence $\phi$ is satisfiable, is decidable.
\end{fact}

\noindent
In the remainder of the section we fix some further definitions and notations,
starting with some useful syntactic abbreviations.

\begin{definition}
Given a list $\vlist{y} = y_1\cdots y_n$ of individual variables, we use the
formula
\[
\arediff{\vlist{y}} \isdef
\bigwedge_{1\leq m < m^{\prime} \leq n} (y_m \foneq y_{m^{\prime}})
\]
to state that the elements $\vlist{y}$ are all distinct.
An \emph{$A$-type} is a formula of the form 
\[
\tau_{S}(x) \isdef
\bigwedge_{a\in S}a(x) \land \bigwedge_{a\in A\setminus S}\lnot a(x).
\]
where $S \subseteq A$.
Here and elsewhere we use the convention that $\bigwedge \nada = \top$ (and 
$\bigvee \nada = \bot$).
The \emph{positive} $A$-type $\tau_{S}^+(x)$ only bears positive information, 
and is defined as 
\[
\tau_{S}^+(x) \isdef \bigwedge_{a\in S}a(x).
\]
Given a one-step model $\osmodel = (D,V)$ we define
\[
\sz{S}_\osmodel \isdef |\{d\in D \mid \osmodel \models \tau_S(d) \}|
\]
as the number of elements of $\osmodel$ that realise the type $\tau_{S}$.
\end{definition}

We often blur the distinction between the formula $\tau_{S}(x)$ and the subset 
$S \subseteq A$, calling $S$ an $A$-type as well.
Note that we have $\osmodel \models \tau_S(d)$ iff $V^{\flat}(d) = S$, so that
we may refer to $V^{\flat}(d)$ as the \emph{type of} $d \in D$ indeed.

\begin{definition}
The quantifier rank $\qr(\varphi)$ of a formula $\varphi \in \ofoei$ (hence also 
for $\ofo$ and $\ofoe$) is defined as follows:
\[\begin{array}{llll}
\qr(\varphi) &\isdef & 0 
  & \text{if $\varphi$ is atomic},
\\ \qr(\neg\psi) & \isdef & \qr(\psi)
\\ \qr(\psi_{1}\mathrel{\hs}\psi_{2}) & \isdef & \max\{\qr(\psi_1),\qr(\psi_2)\}
  & \text{where } \hs \in \{ \land,\lor\}
\\ \qr(Qx.\psi) & \isdef & 1+\qr(\psi),
  & \text{where } Q \in \{\exists,\forall,\qu,\dqu\}
\end{array}\]
Given a monadic logic $\llang$ we write $\osmodel \equiv_k^{\llang} \osmodel'$
to indicate that the models $\osmodel$ and $\osmodel'$ satisfy exactly the same
sentences $\varphi \in \llang$ with $\qr(\varphi) \leq k$. 
We write $\osmodel \equiv^{\llang} \osmodel'$ if $\osmodel \equiv_k^{\llang}
\osmodel'$ for all $k$.
When clear from context, we may omit explicit reference to $\llang$.
\end{definition}

\begin{definition}
A \emph{partial isomorphism} between two models $(D,V )$ and $(D',V ')$ is a 
partial function $f: D \pto D'$ which is injective and satisfies that 
$d \in V (a) \Leftrightarrow f(d) \in V '(a)$ for all $a\in A$ and $ d\in 
\Dom(f)$.
Given two sequences $\vlist{d} \in D^k$ and $\vlist{d'} \in {D'}^k$ we use 
$f:\vlist{d} \mapsto \vlist{d'}$ to denote the partial function $f:D\pto D'$ 
defined as $f(d_i) \isdef d'_i$. 
We will take care to avoid cases where there exist $d_i,d_j$ such that $d_i = 
d_j$ but $d'_i \neq d'_j$.
\end{definition}

Finally, for future reference we briefly discuss the notion of \emph{Boolean 
duals}.
We first give a concrete definition of a dualisation operator on the set of 
monadic formulas.

\begin{definition}\label{def:concreteduals} 
The \emph{(Boolean) dual} $\phi^{\delta} \in {\ofoei}(A)$ of $\phi\in 
{\ofoei}(A)$ is the formula given by:
\begin{align*}
 (a(x))^{\delta} & \isdef  a(x) 
 & (\lnot a(x))^{\delta} & \isdef  \lnot a(x) 
\\ (\top)^{\delta} & \isdef  \bot 
  & (\bot)^{\delta} & \isdef  \top 
\\  (x \approx y)^{\delta} & \isdef  x \not\approx y 
  & (x \not\approx y)^{\delta}& \isdef  x \approx y 
\\ (\phi \wedge \psi)^{\delta} &\isdef  \phi^{\delta} \vee \psi^{\delta} 
  &(\phi \vee \psi)^{\delta}& \isdef  \phi^{\delta} \wedge \psi^{\delta}
\\ (\exists x.\psi)^{\delta} &\isdef  \forall x.\psi^{\delta} 
  &(\forall x.\psi)^{\delta} &\isdef  \exists x.\psi^{\delta} 
\\ (\exists^{\infty} x.\psi)^{\delta} &\isdef \forall^{\infty} x.\psi^{\delta} 
  &(\forall^{\infty} x.\psi)^{\delta} &\isdef  \exists^{\infty} x.\psi^{\delta}
\end{align*}
\end{definition}

\begin{remark}
Where $\llang\in\{\ofo,\ofoe,\ofoei\}$, observe that if $\phi \in \llang(A)$ 
then $\phi^{\delta} \in \llang(A)$. 
Moreover, the operator preserves positivity of the predicates, that is, if $\phi
\in \llang^+(A)$ then $\phi^{\delta} \in \llang^+(A)$.
\end{remark}

\noindent 
The following proposition states that the formulas $\phi$ and $\phi^{\delta}$ 
are Boolean duals.
We omit its proof, which is a routine check.

\begin{proposition}\label{props:duals}
Let $\phi \in \ofoei(A)$ be a monadic formula.
Then $\phi$ and $\phi^{\delta}$ are indeed Boolean duals, in the sense that for
every monadic model $(D,V)$ we have that
\[
(D,V) \models \phi \text{ iff } (D,V^{c}) \not\models \phi^{\delta},
\]
where $V^{c}: A \to \wp (D)$ is the valuation given by $V^{c}(a) \isdef 
D \setminus V(a)$.
\end{proposition}



\section{Normal forms}
\label{sec:normalforms}

In this section we provide normal forms for the logics $\ofo$, $\ofoe$ and 
$\ofoei$. 
These normal forms will be pivotal for characterising the different fragments of 
these logics in later sections.
\begin{convention}
Here and in the sequel it will often be convenient to blur the distinction 
between lists and sets.
For instance, identifying the list $\vlist{T} = T_{1}\cdots T_{n}$ with the
set $\{ T_{1}, \ldots, T_{n} \}$, we may write statements like $S \in \vlist{T}$ 
or $\Pi \subseteq \vlist{T}$. 
Moreover, given a finite set $\Phi = \{\phi_1, \dots, \phi_n\}$, we write $\phi_1 \land \dots \land \phi_n$ as $\bigwedge \Phi$, and $\phi_1 \lor \dots \lor \phi_n$ as $\bigvee \Phi$. If $\Phi$ is empty, we set as usual $\bigwedge \Phi = \top$ and $\bigvee \Phi = \bot$.
Finally, notice that we write $\bigvee_{1\leq m < m^{\prime} \leq n} (y_m \foeq y_{m^{\prime}}) \lor \psi$ as $\arediff{\vlist{y}} \to \psi$. 
\end{convention}


\subsection{Normal form for $\ofo$}

We start by introducing a normal form for monadic first-order logic without 
equality.

\begin{definition}\label{def:bfofo}
Given sets of types $\Sigma, \Pi \subseteq \wp(A)$, we define the following 
formulas:
\[\begin{array}{lll}
   \dgbnfofo{\Sigma}{\Pi} &\isdef&
   \bigwedge_{S\in\Sigma} \exists x. \tau_S(x) \land 
   \forall x. \bigvee_{S\in\Pi} \tau_S(x)
\\ \dbnfofo{\Sigma}      &\isdef& \dgbnfofo{\Sigma}{\Sigma}
\end{array}\]
A sentence of $\ofo(A)$ is in \emph{basic form} if it is 
a disjunction of formulas of the form $\dbnfofo{\Sigma}$.
\end{definition}

Clearly the meaning of the formula $\dbnfofo{\Sigma}$ is that $\Sigma$ is a 
complete description of the collection of types that are realised in a monadic 
model. 
Notice that $  \dgbnfofo{\Sigma}{\Pi} = \dbnfofo{\Sigma} = \forall x. \bot$, for 
$\Sigma =\Pi = \nada$.

\medskip

\noindent
Every $\ofo$-formula is effectively equivalent to a formula in basic form.

\begin{fact}\label{fact:ofonormalform}
There is an effective procedure that transforms an arbitrary $\ofo$-sentence 
$\phi$ into an equivalent formula $\tbas{\phi}$ in basic form.
\end{fact}

This observation is easy to prove using Ehrenfeucht-Fra\"iss\'e games -- proof
sketches can be found in~\cite[Lemma 16.23]{ALG02} 
and~\cite[Proposition 4.14]{Venema2014} --, and the decidability of the satisfiability problem for $\ofo$ (Fact \ref{f:decido}). 
We omit a full proof because it is very similar to the following more complex 
cases.


\subsection{Normal form for $\ofoe$}

When considering a normal form for $\ofoe$, the fact that we can `count types'
using equality yields a more involved basic form.

\begin{definition}
We say that a formula $\phi \in \ofoe(A)$ is in \emph{basic form} if
$\phi = \bigvee \dbnfofoe{\vlist{T}}{\Pi}$ where each disjunct is of the
form
\begin{equation*}
\dbnfofoe{\vlist{T}}{\Pi} = 
\exists \vlist{x}.\big(\arediff{\vlist{x}} \land \bigwedge_i \tau_{T_i}(x_i) 
  \land \forall z.(\arediff{\vlist{x},z} 
  \to \bigvee_{S\in \Pi} \tau_S(z))\big)
\end{equation*}
such that $\vlist{T} \in \wp(A)^k$ for some $k$ and $\Pi \subseteq \vlist{T}$.
\end{definition}

We prove that every sentence of monadic first-order logic with equality is
equivalent to a formula in basic form. 
Although this result seems to be folklore, we provide a detailed proof because 
some of its ingredients will be used later, when we give a normal form for 
$\ofoei$. 
We start by defining the following relation between monadic models.

\begin{definition}
For every $k \in \bbN$ we define the relation $\sim^=_k$ on the class $\umods$
of monadic models by putting
\begin{eqnarray*}
\osmodel \sim^=_k \osmodel' \Longleftrightarrow 
  \forall S\subseteq A \ \big(
    |S|_\osmodel = |S|_{\osmodel'} < k \text{ or } 
    |S|_\osmodel,|S|_{\osmodel'} \geq k 
    \big),
\end{eqnarray*}
where $\osmodel$ and $\osmodel'$ are arbitrary monadic models. 
\end{definition}

Intuitively, two models are related by $\sim^=_k$ when their type information 
coincides `modulo~$k$'.
Later on we prove that this is the same as saying that they cannot be 
distinguished by a sentence of $\ofoe$ with quantifier rank at most $k$. 
As a special case, observe that any two monadic models are related by 
$\sim^{=}_{0}$.

For the moment, we record the following properties of these relations.

\begin{proposition}\label{props:eqrelofoe} The following hold:
\begin{enumerate}
\itemsep 0 pt
\item\label{props:eqrelofoe:i} 
The relation $\sim^=_k$ is an equivalence relation of finite index.
\item\label{props:eqrelofoe:ii} 
Every $E \in \umods/{\sim^=_k}$ is characterised by a sentence $\phi^=_E \in 
\ofoe(A)$ with $\qr(\phi^=_E) = k$.
\end{enumerate}
\end{proposition}

\begin{proof}
We only prove the second statement,
and first we consider the case where $k=0$.
The equivalence relation $\sim^{=}_{0}$ has the class $\umods$ of all monadic 
models as its unique equivalence class, so here we may define $\phi^{=}_{\umods}
\isdef \top$.

From now on we assume that $k>0$.
Let $E \in \umods/{\sim^=_k}$ and let $\osmodel \in E$ be a representative. 
Call $S_1,\dots,S_n \subseteq A$ to the types such that $|S_i|_\osmodel = n_i< k$
and $S'_1,\dots,S'_m \subseteq A$ to those satisfying $|S'_i|_\osmodel \geq k$. 
Note that the union of all the $S_i$ and $S'_i$ yields all the possible 
$A$-types, and that if a type $S_{j}$ is not realised at all, we take $n_j = 0$. 
Now define
\begin{align*}
\phi^=_E \quad \isdef \quad  
   & \bigwedge_{i\leq n} \Big(\exists x_1,\dots,x_{n_i}.
   \arediff{x_1,\dots,x_{n_i}} \ \land \ \bigwedge_{j\leq n_i} \tau_{S_i}(x_j)
\\ & \qquad\qquad\qquad \land 
   \forall z. \arediff{x_1,\dots,x_{n_i},z} \to \lnot\tau_{S_i}(z)\Big)\ 
\\ & \land \bigwedge_{i\leq m} \big(\exists x_1,\dots,x_k.\arediff{x_1,\dots,x_k} \land
    \bigwedge_{j\leq k} \tau_{S'_i}(x_j) \big),
\end{align*}
where we understand that any conjunct of the form $\exists x_1,
\dots,x_{l}.\psi$ with $l = 0$ is simply omitted (or, to the same effect, 
defined as $\top$).
It is easy to see that $\qr(\phi^=_E) = k$ and that $\osmodel' \models 
\phi^=_E$ iff $\osmodel' \in E$. 
Intuitively, $\phi^=_E$ gives a specification of $E$ ``type by type''; 
in particular observe that $\phi^=_\emodel \equiv \forall x. \bot$.
\end{proof}

Next we recall a (standard) notion of Ehrenfeucht-Fra\"{\i}ss\'e game for 
$\ofoe$ which will be used to establish the connection between ${\sim^=_k}$ and 
$\equiv_k^{\ofoe}$.

\begin{definition}
Let $\osmodel_0 = (D_0,V_0)$ and $\osmodel_1 = (D_1,V_1)$ be monadic models. 
We define the game $\efgame^=_k(\osmodel_0,\osmodel_1)$ between \abelard and 
\eloise.
If $\osmodel_i$ is one of the models we use $\osmodel_{-i}$ to denote the other
model. 
A position in this game is a pair of sequences $\vlist{s_0} \in D_0^n$ and 
$\vlist{s_1} \in D_1^n$ with $n \leq k$. 
The game consists of $k$ rounds where in round $n+1$ the following steps are 
made:
\begin{enumerate}[1.]
\itemsep 0 pt \parsep 0 pt
\item \abelard chooses an element $d_i$ in one of the $\osmodel_i$;
\item \eloise responds with an element $d_{-i}$ in the model $\osmodel_{-i}$.
\end{enumerate}
In this way, the sequences $\vlist{s_i} \in D_i^n$ of elements chosen up to 
round $n$ are extended to ${\vlist{s_i}' \isdef  \vlist{s_i}\cdot d_i}$. 
Player \eloise survives the round iff she does not get stuck and the function
$f_{n+1}: \vlist{s_0}' \mapsto \vlist{s_1}'$ is a partial isomorphism of monadic 
models. 
Finally, player \eloise wins the match iff she survives all $k$ rounds.

Given $n\leq k$ and $\vlist{s_i} \in D_i^n$ such that $f_n:\vlist{s_0}\mapsto
\vlist{s_1}$ is a partial isomorphism, we write $\efgame_{k}^=(\osmodel_0,
\osmodel_1)@(\vlist{s_0},\vlist{s_1})$ to denote the (initialised) game where 
$n$ moves have been played and $k-n$ moves are left to be played.
\end{definition}

\begin{proposition}\label{prop:connofoe}
The following are equivalent:
\begin{enumerate}
\itemsep 0 pt \parsep 0 pt
\item\label{prop:connofoe:i} $\osmodel_0 \equiv_k^{\ofoe} \osmodel_1$,
\item\label{prop:connofoe:ii} $\osmodel_0 \sim_k^= \osmodel_1$,
\item\label{prop:connofoe:iii} \eloise has a winning strategy in 
   $\efgame_k^=(\osmodel_0,\osmodel_1)$.
	\end{enumerate}
\end{proposition}
\begin{proof}
Step~(\ref{prop:connofoe:i}) to~(\ref{prop:connofoe:ii}) is direct by
Proposition~\ref{props:eqrelofoe}. 
For~(\ref{prop:connofoe:ii}) to~(\ref{prop:connofoe:iii}) we give a winning 
strategy for \eloise in $\efgame_k^=(\osmodel_0,\osmodel_1)$ by showing the 
following claim.
\begin{claimfirst}
Let $\osmodel_0 \sim_k^= \osmodel_1$ and $\vlist{s_i} \in D_i^n$ be such that
$n<k$ and $f_n:\vlist{s_0}\mapsto\vlist{s_1}$ is a partial isomorphism; then 
\eloise can survive one more round in 
$\efgame_{k}^=(\osmodel_0,\osmodel_1)@(\vlist{s_0},\vlist{s_1})$.
\end{claimfirst}

\begin{pfclaim}
Let \abelard pick $d_i\in D_i$ such that the type of $d_i$ is $T \subseteq A$. 
If $d_i$ had already been played then \eloise picks the same element as before
and $f_{n+1} = f_n$.
If $d_i$ is new and $|T|_{\osmodel_i} \geq k$ then, as at most $n<k$ elements 
have been played, there is always some new $d_{-i} \in D_{-i}$ that \eloise can 
choose to match $d_i$. 
If $|T|_{\osmodel_i} = m < k$ then we know that $|T|_{\osmodel_{-i}} = m$. 
Therefore, as $d_i$ is new and $f_n$ is injective, there must be a $d_{-i} \in 
D_{-i}$ that \eloise can choose. 
\end{pfclaim}
	
Step~(\ref{prop:connofoe:iii}) to~(\ref{prop:connofoe:i}) is a standard
result~\cite[Corollary 2.2.9]{fmt} which we prove anyway because we will need 
to extend it later. 
We prove the following loaded statement.
\begin{claim}
Let $\vlist{s_i} \in D_i^n$ and $\phi(z_1,\dots,z_n) \in \ofoe(A)$ be such 
that $\qr(\phi) \leq k-n$. 
If \eloise has a winning strategy in the game 
$\efgame_k^=(\osmodel_0,\osmodel_1)@(\vlist{s_0},\vlist{s_1})$ then 
$\osmodel_0 \models \phi(\vlist{s_0})$ iff 
$\osmodel_1 \models \phi(\vlist{s_1})$.
\end{claim}
	
\begin{pfclaim}
If $\phi$ is atomic the claim holds because of $f_n:\vlist{s_0}\mapsto 
\vlist{s_1}$ being a partial isomorphism. 
The Boolean cases are straightforward.
Let $\phi(z_1,\dots,z_n) = \exists x. \psi(z_1,\dots,z_n,x)$ and suppose
$\osmodel_0 \models \phi(\vlist{s_0})$. 
Hence, there exists $d_0 \in D_0$ such that $\osmodel_0 \models 
\psi(\vlist{s_0},d_0)$.
By hypothesis we know that \eloise has a winning strategy for 
$\efgame_k^=(\osmodel_0,\osmodel_1)@(\vlist{s_0},\vlist{s_1})$. 
Therefore, if \abelard picks $d_0\in D_0$ she can respond with some $d_1\in D_1$ 
and have a winning strategy for 
$\efgame_{k}^=(\osmodel_0,\osmodel_1)%
@(\vlist{s_0}{\cdot}d_0,\vlist{s_1}{\cdot}d_1)$.
By induction hypothesis, because $\qr(\psi) \leq k- (n+1)$, we have that 
$\osmodel_0 \models \psi(\vlist{s_0},d_0)$ iff $\osmodel_1 \models 
\psi(\vlist{s_1},d_1)$ and hence 
$\osmodel_1 \models \exists x.\psi(\vlist{s_1},x)$. 
The opposite direction is proved by a symmetric argument. 
\end{pfclaim}

\noindent
We finish the proof of the proposition by combining these two claims.
\end{proof}

\begin{theorem}
\label{thm:bnfofoe}
There is an effective procedure that transforms an arbitrary $\ofoe$-sentence 
$\phi$ into an equivalent formula $\tbas{\phi}$ in basic form.
\end{theorem}

\begin{proof}
Let $\qr(\psi) = k$ and let $\ext{\psi}$ be the class of models satisfying 
$\psi$.
As $\umods/{\equiv_k^{\ofoe}}$ is the same as $\umods/{\sim_k^=}$ by
Proposition~\ref{prop:connofoe}, it is easy to see that $\psi$ is equivalent to 
$\bigvee \{ \phi^=_E \mid E \in \ext{\psi}/{\sim_k^=} \}$.
Now it only remains to see that each $\phi^=_E$ is equivalent to the sentence
$\dbnfofoe{\vlist{T}}{\Pi}$ for some $\vlist{T},\Pi \subseteq \wp(A)$ with
$\Pi \subseteq \vlist{T}$.

The crucial observation is that we will use $\vlist{T}$ and $\Pi$ to give a 
specification of the types ``element by element''. 
Let $\osmodel$ be a representative of the equivalence class $E$. 
Call $S_1,\dots,S_n \subseteq A$ to the types such that $|S_i|_\osmodel = n_i 
< k$ and $S'_1,\dots,S'_m \subseteq A$ to those satisfying $|S'_j|_\osmodel 
\geq k$. 
The size of the sequence $\vlist{T}$ is defined to be $(\sum_{i=1}^n n_i) + 
k\times m$ where $\vlist{T}$ contains exactly $n_i$ occurrences of type $S_i$ 
and at least $k$ occurrences of each $S'_j$.
On the other hand we set $\Pi \isdef \{S'_1,\dots,S'_m\}$. 
It is straightforward to check that $\Pi \subseteq \vlist{T}$ and $\phi^=_E$ is
equivalent to $\dbnfofoe{\vlist{T}}{\Pi}$. (Observe however, that the quantifier rank of the latter is only bounded by
$k\times 2^{|A|} + 1$.) 
In particular $\phi^=_\emodel \equiv \dbnfofoe{\nada}{\nada} = \forall x. \bot$.

The effectiveness of the procedure hence follows from the fact that, given the previous bound on the size of a normal form, it is possible to
non-deterministically guess the number of disjuncts, types and associated parameters for each conjunct and repeatedly check whether the formulas $\phi$
and    $\bigvee \dbnfofoe{\vlist{T}}{\Pi}$ are equivalent,  this latter problem being decidable by Fact \ref{f:decido}.
\end{proof}


\subsection{Normal form for $\ofoei$}

The logic $\ofoei$ extends $\ofoe$ with the capacity to tear apart finite and 
infinite sets of elements. 
This is reflected in the normal form for $\ofoei$ by adding extra information
to the normal form of $\ofoe$.

\begin{definition}\label{def:basicform_fofoei}
We say that a formula $\phi \in \ofoei(A)$ is in \emph{basic form} if 
$\phi = \bigvee \dbnfofoei{\vlist{T}}{\Pi}{\Sigma}$ where each disjunct is 
of the form
\[
\dbnfofoei{\vlist{T}}{\Pi}{\Sigma} \isdef \dbnfofoe{\vlist{T}}{\Pi \cup \Sigma} 
  \land \dbnfinf{\Sigma}
\]
where
\[
\dbnfinf{\Sigma} \isdef  
\bigwedge_{S\in\Sigma} \qu y.\tau_S(y) \land 
  \dqu y.\bigvee_{S\in\Sigma} \tau_S(y).
\]
Here $\vlist{T} \in \wp(A)^{k}$ for some $k$, and $\Pi,\Sigma \subseteq \wp(A)$ 
are such that $\Sigma \cup \Pi \subseteq \vlist{T}$.
\end{definition}

Intuitively, the formula $\dbnfinf{\Sigma}$ says that (1) for every type $S\in
\Sigma$, there are infinitely many elements satisfying $S$ and (2) only finitely
many elements do not satisfy any type in $\Sigma$.
As a special case, the formula $\dbnfinf{\nada}$ expresses that the model is 
finite.
A short argument reveals that, intuitively, every disjunct of the form
$\dbnfofoei{\vlist{T}}{\Pi}{\Sigma}$ expresses that any monadic model satisfying
it admits a partition of its domain in three parts:
\begin{enumerate}[(i)]
\itemsep 0 pt
\item distinct elements $t_1,\dots,t_n$ with respective types $T_1,\dots,T_n$,
\item finitely many elements whose types belong to $\Pi$, and
\item for each $S\in \Sigma$, infinitely many elements with type $S$.
\end{enumerate}
Observe that basic formulas of $\ofoe$ are \emph{not} basic formulas of 
$\ofoei$.

In the same way as before, we define an equivalence relation $\sim^\infty_k$ 
on monadic models which refines $\sim^=_{k}$ by adding information about the 
(in-)finiteness of the types.

\begin{definition}
For every $k \in \bbN$ we define the relation $\sim^{\infty}_k$ on the class 
$\umods$ of monadic models by putting
\begin{eqnarray*}
  \osmodel \sim^\infty_0 \osmodel' 
  & \Longleftrightarrow 
  & \text{always}
\\\osmodel \sim^\infty_{k+1} \osmodel' 
  & \Longleftrightarrow 
  & \forall S\subseteq A \ \big(
    |S|_\osmodel = |S|_{\osmodel'} < k 
    \text{ or }     k \leq |S|_\osmodel,|S|_{\osmodel'} < \omega
     \text{ or }    |S|_\osmodel,|S|_{\osmodel'} \geq \omega 
    \big),
\end{eqnarray*}
where $\osmodel$ and $\osmodel'$ are arbitrary monadic models. 
\end{definition}

\begin{proposition}\label{props:eqrelolque} The following hold:
\begin{enumerate}
\itemsep 0 pt
\item 
The relation $\sim^\infty_k$ is an equivalence relation of finite index.
\item 
The relation $\sim^\infty_k$ is a refinement of $\sim^=_k$.
\item 
Every $E \in \umods/{\sim^\infty_k}$ is characterised by a sentence
   $\phi^\infty_E \in \ofoei(A)$ with $\qr(\phi) = k$.
	\end{enumerate}
\end{proposition}

\begin{proof}
We only prove the last point, for $k>0$. 
Let $E \in \umods/{\sim^\infty_k}$ and let $\osmodel \in E$ be a representative 
of the class. 
Let $E' \in \umods/{\sim^=_k}$ be the equivalence class of $\osmodel$ with 
respect to $\sim^=_k$.
Let $S_1,\dots,S_n \subseteq A$ be all the types such that $|S_i|_\osmodel \geq 
\omega$, and define
\[
\phi^\infty_E \isdef  \phi^=_{E'} \land \dbnfinf{\{S_1,\dots,S_n\}} .
\]
It is not difficult to see that $\qr(\phi^\infty_E) = k$ and that $\osmodel'
\models \phi^\infty_E$ iff $\osmodel' \in E$.  
In particular $\phi^\infty_\emodel \equiv \dbnfofoei{\nada}{\nada}{\nada} =
\forall x. \bot \land \dqu y. \bot$.
\end{proof}

Now we give a version of the Ehrenfeucht-Fra\"ss\'e game for $\ofoei$. 
This game, which extends $\efgame^=_k$ with moves for $\qu$, is the
adaptation of the Ehrenfeucht-Fra\"{\i}ss\'e game for monotone generalised 
quantifiers found in~\cite{krawczyk1976ehrenfeucht} to the case of full monadic 
first-order logic. 

\begin{definition}
Let $\osmodel_0 = (D_0,V_0)$ and $\osmodel_1 = (D_1,V_1)$ be monadic models. 
We define the game $\efgame^\infty_k(\osmodel_0,\osmodel_1)$ between \abelard
and \eloise. 
A position in this game is a pair of sequences $\vlist{s_0} \in D_0^n$ and
$\vlist{s_1} \in D_1^n$ with $n \leq k$. 
The game consists of $k$ rounds, where in round $n+1$ the following steps are 
made. First \abelard chooses to perform one of the following types of moves:
\begin{enumerate}[(a)]
\itemsep 0 pt \parsep 0 pt \topsep 0 pt \parskip 0 pt%
\item second-order move:
\begin{enumerate}[1.]
\itemsep 0 pt \parsep 0 pt \topsep 0 pt \parskip 0 pt \partopsep 0 pt%
\item \abelard chooses an infinite set $X_i \subseteq D_i$;
\item \eloise responds with an infinite set $X_{-i} \subseteq D_{-i}$;
\item \abelard chooses an element $d_{-i} \in X_{-i}$;
\item \eloise responds with an element $d_i \in X_i$.
\end{enumerate}
\item first-order move:
\begin{enumerate}[1.]
\itemsep 0 pt
\parsep 0 pt
\item \abelard chooses an element $d_i \in D_i$;
\item \eloise responds with an element $d_{-i} \in D_{-i}$.
\end{enumerate}
\end{enumerate}
The sequences $\vlist{s_i} \in D_i^n$ of elements chosen up to round $n$ are 
then extended to ${\vlist{s_i}' \isdef  \vlist{s_i}\cdot d_i}$. 
\eloise survives the round iff she does not get stuck and the function $f_{n+1}:
\vlist{s_0}' \mapsto \vlist{s_1}'$ is a partial isomorphism of monadic models.
\end{definition}

\begin{proposition}\label{prop:connolque}
The following are equivalent:
\begin{enumerate}
\itemsep 0 pt
\item\label{prop:connolque:i} 
$\osmodel_0 \equiv_k^{\ofoei} \osmodel_1$,
\item\label{prop:connolque:ii}
$\osmodel_0 \sim_k^\infty \osmodel_1$,
\item\label{prop:connolque:iii}
\eloise has a winning strategy in $\efgame_k^\infty(\osmodel_0,\osmodel_1)$.
\end{enumerate}
\end{proposition}

\begin{proof}
Step~(\ref{prop:connolque:i}) to~(\ref{prop:connolque:ii}) is direct by 
Proposition~\ref{props:eqrelolque}. For~(\ref{prop:connolque:ii}) 
to~(\ref{prop:connolque:iii}) we show the following.

\begin{claimfirst}
Let $\osmodel_0 \sim_k^\infty \osmodel_1$ and $\vlist{s_i} \in D_i^n$ be such
that $n<k$ and $f_n:\vlist{s_0}\mapsto\vlist{s_1}$ is a partial isomorphism.
Then \eloise can survive one more round in $\efgame_{k}^\infty(\osmodel_0,
\osmodel_1)@(\vlist{s_0},\vlist{s_1})$.
\end{claimfirst}
\begin{pfclaim}
We focus on the second-order moves because the first-order moves are the same as
in the corresponding Claim of Proposition~\ref{prop:connofoe}. 
Let \abelard choose an infinite set $X_i \subseteq D_i$, we would like \eloise
to choose an infinite set $X_{-i} \subseteq D_{-i}$ such that the following 
conditions hold:

\begin{enumerate}[(a)]
\parskip 0pt
\item\label{it:piso} 
The map $f_n$ is a well-defined partial isomorphism between the restricted
monadic models $\osmodel_0{\rest}X_0$ and $\osmodel_1{\rest}X_1$,

\item\label{it:equiv}
For every type $S$ there is an element $d\in X_i$ of type $S$ which is 
\emph{not} connected by $f_n$ iff there is such an element in $X_{-i}$.
\end{enumerate}	
\begin{figure}[ht]
\centering
\includegraphics[scale=0.7]{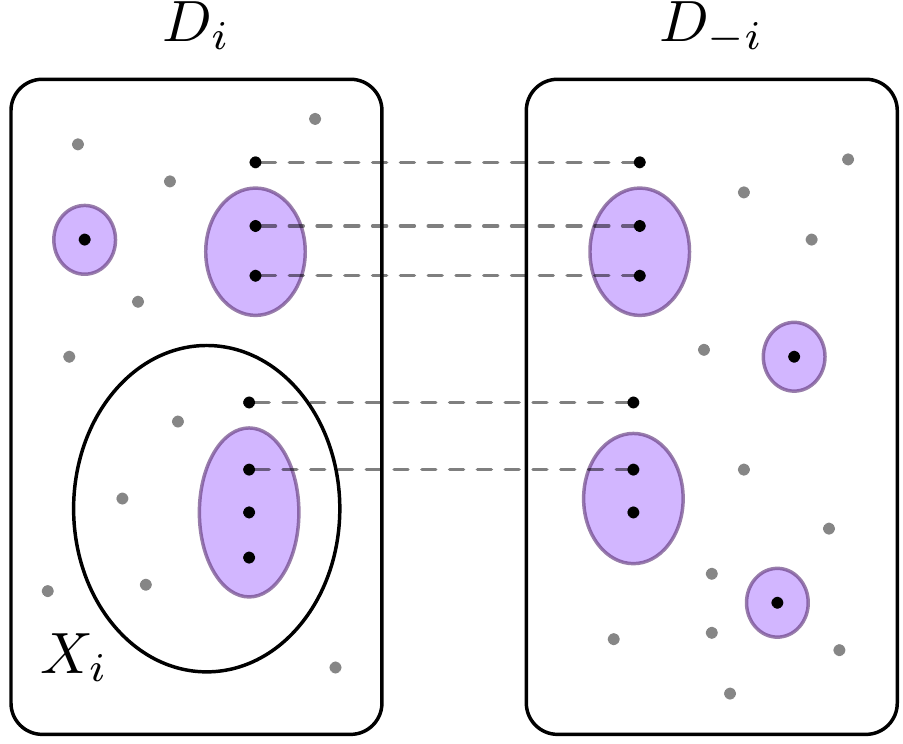}
\caption{Elements of type $S$ have coloured background.}
\label{fig:efinf}
\end{figure}

First we prove that such a set $X_{-i}$ exists. 
To satisfy item~\eqref{it:piso} \eloise just needs to add to $X_{-i}$ the 
elements connected to $X_i$ by $f_n$; this is not a problem.

For item~\eqref{it:equiv} we proceed as follows: for every type $S$ such that 
there is an element $d\in X_i$ of type $S$, we add a new element $d'\in D_{-i}$
of type $S$ to $X_{-i}$. 
To see that this is always possible, observe first that $\osmodel_0 
\sim_k^\infty \osmodel_1$ implies $\osmodel_0 \sim_k^= \osmodel_1$. 
Using the properties of this relation, we divide in two cases:
\begin{itemize}
\item If $|S|_{D_i} \geq k$ we know that $|S|_{D_{-i}} \geq k$ as well. 
From the elements of $D_{-i}$ of type $S$, at most $n<k$ are used by $f_n$. 
Hence, there is at least one $d'\in D_{-i}$ of type $S$ to choose from.

\item If $|S|_{D_i} < k$ we know that $|S|_{D_{i}} = |S|_{D_{-i}}$. 
From the elements of $D_{i}$ of type $S$, at most $|S|_{D_{i}}-1$ are used by 
$f_n$. 
(The reason for the `$-1$' is that we are assuming that we have just chosen a 
$d\in X_i$ which is not in $f_n$.) 
Using that $|S|_{D_{i}} = |S|_{D_{-i}}$ and that $f_n$ is a partial isomorphism 
we can again conclude that there is at least one $d'\in D_{-i}$ of type $S$ to 
choose from.
\end{itemize}
	
Finally, we need to show that \eloise can choose $X_{-i}$ to be infinite.
To see this, observe that $X_{i}$ is infinite, while there are only finitely 
many types.
Hence there must be some $S$ such that $|S|_{X_i} \geq \omega$. 
It is then safe to add infinitely many elements for $S$ in $X_{-i}$ while 
considering point~\eqref{it:equiv}. 
Moreover, the existence of infinitely many elements satisfying $S$ in $D_{-i}$
is guaranteed by $\osmodel_0 \sim_k^\infty \osmodel_1$.

Having shown that \eloise can choose a set $X_{-i}$ satisfying the above 
conditions, it is now clear that using point~\eqref{it:equiv} \eloise can 
survive the ``first-order part'' of the second-order move we were considering.
This finishes the proof of the claim.
\end{pfclaim}

\noindent
Returning to the proof of Proposition~\ref{prop:connolque}, for 
step~(\ref{prop:connolque:iii}) to~(\ref{prop:connolque:i}) we prove the following.
	
\begin{claim}
Let $\vlist{s_i} \in D_i^n$ and $\phi(z_1,\dots,z_n) \in \ofoei(A)$ be such
that $\qr(\phi) \leq k-n$. 
If \eloise has a winning strategy in 
$\efgame_k^\infty(\osmodel_0,\osmodel_1)@(\vlist{s_0},\vlist{s_1})$ then 
$\osmodel_0 \models \phi(\vlist{s_0})$ iff 
$\osmodel_1 \models \phi(\vlist{s_1})$.
\end{claim}
	
\begin{pfclaim}
All the cases involving operators of $\ofoe$ are the same as in 
Proposition~\ref{prop:connofoe}. 
We prove the inductive case for the generalised quantifier. 
Let $\phi(z_1,\dots,z_n)$ be of the form $\qu x.\psi(z_1,\dots,z_n,x)$ and 
let $\osmodel_0 \models \phi(\vlist{s_0})$. 
Hence, the set $X_{0} \isdef \{ d_{0} \in D_{0} \mid \osmodel_0 \models 
\psi(\vlist{s_0},d_0) \}$ is infinite.

By assumption \eloise has a winning strategy in 
$\efgame_k^\infty(\osmodel_0,\osmodel_1)@(\vlist{s_0},\vlist{s_1})$.
Therefore, if \abelard plays a second-order move by picking $X_0 \subseteq D_0$
she can respond with some infinite set $X_1 \subseteq D_1$. 
We claim that 
$\osmodel_1 \models \psi(\vlist{s_1},d_1)$ for every $d_1\in X_1$. 
First observe that if this holds then the set $X'_1 \isdef  \{ d_1 \in D_1 \mid 
\osmodel_1 \models \psi(\vlist{s_1},d_1)\}$ must be infinite, and hence 
$\osmodel_1 \models \qu x.\psi(\vlist{s_1},x)$.

Assume, for a contradiction, that $\osmodel_1 \not\models \psi(\vlist{s_1},d'_1)$
for some $d'_1\in X_1$. 
Let \abelard play this $d'_1$ as the second part of his move. 
Then, as \eloise has a winning strategy, she will respond with some $d'_0 \in 
X_0$ for which she has a winning strategy in 
$\efgame_{k}^\infty(\osmodel_0,\osmodel_1)%
  @(\vlist{s_0}{\cdot}d'_0,\vlist{s_1}{\cdot}d'_1)$. 
But then by our induction hypothesis, which applies since $\qr(\psi) \leq
k-(n+1)$, we may infer from $\osmodel_1 \not\models \psi(\vlist{s_1},d'_1)$
that $\osmodel_0 \not\models \psi(\vlist{s_0},d'_0)$.
This clearly contradicts the fact that $d'_{0} \in X_{0}$.
\end{pfclaim}
	
\noindent
Combining the claims finishes the proof of the proposition.
\end{proof}

\begin{theorem}
\label{thm:bfofoei}
There is an effective procedure that transforms an arbitrary $\ofoei$-sentence 
$\phi$ into an equivalent formula $\tbas{\phi}$ in basic form.
\end{theorem}

\begin{proof}
This can be proved using the same argument as in Theorem~\ref{thm:bnfofoe} but
based on Proposition \ref{prop:connolque}. 
Hence we only focus on showing that $\phi_E^\infty \equiv 
\dbnfofoei{\vlist{T}}{\Pi}{\Sigma}$ for some 
$\vlist{T},\Pi,\Sigma \subseteq \wp(A)$ such that $\Sigma \cup \Pi \subseteq
\vlist{T}$, where $\phi_E^\infty$ is the sentence characterising
$E \in \umods/{\sim^\infty_k}$ from 
Proposition \ref{props:eqrelolque}(\ref{props:eqrelofoe:ii}). 
Recall that
\[
\phi^\infty_E = \phi^=_{E'} \land \dbnfinf{\Sigma}
\]
where $\Sigma$ is the collection of types that are realised by infinitely many 
elements.
Using Theorem~\ref{thm:bnfofoe} on $\phi^=_{E'}$ we know that this is 
equivalent to
\[
\phi^\infty_E = \dbnfofoe{\vlist{T}}{\Pi'} \land \dbnfinf{\Sigma}
\]
where $\Pi' \subseteq \vlist{T}$.
Observe that we may assume that $\Sigma \subseteq \Pi$, otherwise the formula 
would be inconsistent.
Now separate $\Pi'$ as $\Pi' = \Pi \uplus \Sigma$ where $\Pi \isdef 
\Pi'\setminus\Sigma$ consists of the types that are satisfied by finitely many 
elements.
Then we find 
\[
\phi^\infty_E \equiv 
\dbnfofoe{\vlist{T}}{\Pi\cup\Sigma} \land \dbnfinf{\Sigma}.
\]
Therefore, we can conclude that $\phi^\infty_E \equiv 
\dbnfofoei{\vlist{T}}{\Pi}{\Sigma}$.
\end{proof}

\noindent
The following slightly stronger normal form will be useful in later 
chapters.

\begin{proposition}\label{prop:bfofoei-sigmapi}
For every sentence in the basic form $\bigvee \dbnfofoei{\vlist{T}}{\Pi}{\Sigma}$
it is possible to assume, without loss of generality, that $\Sigma \subseteq 
\Pi \subseteq \vlist{T}$.
\end{proposition}
\begin{proof}
This is direct from observing that $\dbnfofoei{\vlist{T}}{\Pi}{\Sigma}$ is 
equivalent to $\dbnfofoei{\vlist{T}}{\Pi\cup\Sigma}{\Sigma}$. 
To check it we just unravel the definitions and observe that
$\dbnfofoe{\vlist{T}}{\Pi \cup \Sigma} \land \dbnfinf{\Sigma}$ is equivalent to
$\dbnfofoe{\vlist{T}}{\Pi \cup \Sigma \cup \Sigma} \land \dbnfinf{\Sigma}$.
\end{proof}



\section{Monotonicity}
\label{sec-mono}

In this section we provide our first characterisation result, which concerns the
notion of monotonicity. 

\begin{definition}\label{def:mono}
Let $V$ and $V'$ be two valuations on the same domain $D$, then we say that $V'$
is a \emph{$B$-extension of} $V$, notation: $V \leq_{B} V'$, if $V(b) \subseteq
V'(b)$ for every $b \in B$, and $V(a) = V'(a)$ for every $a \in A \setminus B$.

Given a monadic logic $\llang$ and a formula $\phi \in \llang(A)$ we say 
that $\phi$ is \emph{monotone in $B \subseteq A$} if 
\begin{equation}
\label{eq:mono}
(D,V),g \models \phi \text{ and } V \leq_{B} V' \text{ imply } 
(D,V'),g \models \phi,
\end{equation}
for every pair of monadic models $(D,V)$ and $(D,V')$ and every assignment 
$g:\fovar\to D$.
\end{definition}

\begin{remark}\label{rem:monotprodeach}
It is easy to prove that a formula is monotone in $B \subseteq A$ if and only if 
it is monotone in every $b \in B$. 
\end{remark}

The semantic property of monotonicity can usually be linked to the syntactic 
notion of positivity.
Indeed, for many logics, a formula $\phi$ is monotone in $a \in A$ iff 
$\phi$ is equivalent to a formula where all occurrences of $a$ have a positive
polarity, that is, they are situated in the scope of an even number of 
negations.

\begin{definition}
For $\llang \in \{ \ofo, \ofoe \}$ we define the fragment of $A$-formulas that 
are \emph{positive} in all predicates in $B$, in short: the \emph{$B$-positive 
formulas} by the following grammar:
\[
\phi \defbnf  \psi \mid b(x) 
  \mid \phi \land \phi \mid \phi \lor \phi
  \mid \exists x.\phi \mid \forall x.\phi, 
\]
where $b \in B$ and $\psi \in \llang(A\setminus B)$ (that is, there are no 
occurrences of any $b \in B$ in $\psi$).
Similarly, the $B$-positive fragment of $\ofoei$ is given by
\[
\phi \defbnf  \psi \mid b(x) 
  \mid \phi \land \phi \mid \phi \lor \phi
  \mid \exists x.\phi \mid \forall x.\phi 
  \mid \qu x.\phi \mid \dqu x.\phi, 
\]
where $b\in B$ and $\psi \in \ofoei(A\setminus B)$. 

In all three cases, we let $\monot{\llang(A)}{B}$ denote the set of $B$-positive
sentences.
\end{definition}

Note that the difference between the fragments $\monot{\ofo(A)}{B}$ and 
$\monot{\ofoe(A)}{B}$ lies in the fact that in the latter case, the `$B$-free'
formulas $\psi$ may contain the equality symbol. 
Clearly $\monot{\llang(A)}{A}= \llang^+$.

\begin{theorem}
\label{t:mono}
Let $\phi$ be a sentence of the monadic logic $\llang(A)$, where $\llang \in 
\{ \ofo, \ofoe, \ofoei \}$.
Then $\phi$ is monotone in a set $B \subseteq A$ if and only if there is a 
equivalent formula $\phi^{\tmono} \in \monot{\llang(A)}{B}$.
Furthermore, it is decidable whether a sentence $\phi \in \llang(A)$ has this 
property or not.
\end{theorem}

The `easy' direction of the first claim of the theorem is taken care of by the following 
proposition.

\begin{proposition}
\label{p:monoismonot}
Every formula $\phi \in \monot{\llang(A)}{B}$ is monotone in $B$,
where $\llang$ is one of the logics $\{ \ofo, \ofoe, \ofoei \}$.
\end{proposition}

\begin{proof} 
The case for $D= \nada$ being immediate, we assume $D \neq \nada$.
The proof is a routine argument by induction on the complexity of $\phi$.
That is, we show by induction, that any formula $\phi$ in the $B$-positive
fragment (which may not be a sentence) satisfies \eqref{eq:mono}, for every 
monadic model $(D,V)$, valuation $V' \geq_{B} V$ and assignment ${g:\fovar\to 
D}$.
We focus on the generalised quantifiers. 
Let $(D,V),g \models \phi$ and
$V \leq_{B} V'$.
\begin{enumerate}[\textbullet]
\item
Case $\phi = \qu x.\phi'(x)$. By definition there exists an infinite set
$I\subseteq D$ such that for all $d\in I$ we have $(D,V),g[x\mapsto d] \models 
\phi'(x)$. 
By induction hypothesis $(D,V'),g[x\mapsto d] \models\phi'(x)$ 
for all $d \in I$. Therefore $(D,V'),g \models \qu x.\phi'(x)$.

\item Case $\phi = \dqu x.\phi'(x)$. 
Hence there exists $C\subseteq D$ such that for all $d\in C$ we have 
$(D,V),g[x\mapsto d] \models \phi'(x)$ and $D\setminus C$ is \emph{finite}. 
By induction hypothesis $(D,V'),g[x\mapsto d] \models \phi'(x)$ for all $d 
\in C$. 
Therefore $(D,V'),g \models \dqu x.\phi'(x)$.
\end{enumerate}
This finishes the proof.
\end{proof}

The `hard' direction of the first claim of the theorem states that the fragment
$\monot{\ofo}{B}$ is complete for monotonicity in $B$. In order to prove it,
we need to show that every sentence which is monotone in $B$ is equivalent to 
some formula in $\monot{\ofo}{B}$. 
We actually are going to prove a stronger result.

\begin{proposition}
\label{p:efftrans}
Let $\llang$ be one of the logics $\{ \ofo, \ofoe, \ofoei \}$.
There exists an effective translation $(-)^\tmono:\llang(A) \to \monot{\llang(A)}{B}$ such that
a sentence ${\phi \in \llang(A)}$ is monotone in $B \subseteq A$ 
only if 
$\phi\equiv \phi^\tmono$.
\end{proposition}

We prove the three manifestations of Proposition \ref{p:efftrans} separately,
in three respective subsections.

\begin{proofof}{Theorem \ref{t:mono}}
The first claim of the Theorem is an immediate consequence of
Proposition~\ref{p:efftrans}.
By effectiveness of the translation and Fact \ref{f:decido}, it is therefore 
decidable whether a sentence $\phi \in \llang(A)$ is monotone in $B \subseteq A$
or not.
\end{proofof}

The following definition will be used throughout in the remaining of the section.

\begin{definition}
Given $S \subseteq A$ and $B \subseteq A$ we use the following notation
\[
\tau^{B}_S(x) \isdef  \bigwedge_{b\in S} b(x) \land 
   \bigwedge_{b\in A\setminus (S\cup B)}\lnot b(x) ,
\]
for what we call the \emph{$B$-positive} $A$-type $\tau^{B}_S$.
\end{definition}

Intuitively, $\tau^{B}_S$ works almost like the $A$-type $\tau_S$, the 
difference being that $\tau^{B}_S$ discards the negative information for the 
names in $B$.
If $B = \{a\}$ we write $\tau^a_S$ instead of $\tau^{\{a\}}_S$. 
Observe that with this notation, $\tau^+_S$ is equivalent to $\tau^A_S$.

\subsection{Monotone fragment of $\ofo$}

In this subsection we prove the $\ofo$-variant of Proposition~\ref{p:efftrans}.
That is, we give a translation that constructively maps arbitrary sentences
into $\monot{\ofo}{B}$ and that moreover it preserves truth iff the given
sentence is monotone in $B$.
To formulate the translation we need to introduce some new notation. 

\begin{definition}\label{def:monbasicformofoe}
Let $B\subseteq A$ be a finite set of names. 
The $B$-positive variant of $\dbnfofo{\Sigma}$ is given as follows:
\[
\mondbnfofo{\Sigma}{B} \isdef  
\bigwedge_{S\in\Sigma} \exists x. \tau^{B}_S(x) \land 
  \forall x. \bigvee_{S\in\Sigma} \tau^{B}_S(x).
\]
We also introduce the following generalised forms of the above notation:
\[
\mondgbnfofo{\Sigma}{\Pi}{B} \isdef  
\bigwedge_{S\in\Sigma} \exists x. \tau^{B}_S(x) \land 
  \forall x. \bigvee_{S\in\Pi} \tau^{B}_S(x).
\]
The \emph{positive} variants of the above notations are defined as 
$\posdbnfofo{\Sigma} \isdef  \mondbnfofo{\Sigma}{A}$ and 
$\posdgbnfofo{\Sigma}{\Pi} \isdef  \mondgbnfofo{\Sigma}{\Pi}{A}$.
\end{definition}

\begin{proposition}
\label{p:fomon}
There exists an effective translation $(-)^\tmono:\ofo(A) \to \monot{\ofo(A)}{B}$ such that
a sentence ${\phi \in \ofo(A)}$ is monotone in $B\subseteq A$ if and only if 
$\phi\equiv \phi^\tmono$.
\end{proposition}
\begin{proof}
To define the translation, by Fact \ref{fact:ofonormalform}, we assume, without loss of generality, that $\phi$
is in the normal form $\bigvee \dbnfofo{\Sigma}$ given in 
Definition~\ref{def:bfofo}, where
$\dbnfofo{\Sigma} = 
\bigwedge_{S\in\Sigma} \exists x. \tau_S(x) \land 
  \forall x. \bigvee_{S\in\Sigma} \tau_S(x)$.
We define the translation as
\[
(\bigvee \dbnfofo{\Sigma})^\tmono\isdef  \bigvee \mondbnfofo{\Sigma}{B}.
\]
From the construction it is clear that $\phi^\tmono \in \monot{\ofo(A)}{B}$ 
and therefore the right-to-left direction of the proposition is immediate by 
Proposition~\ref{p:monoismonot}. 
For the left-to-right direction assume that $\phi$ is monotone in $B$, we 
have to prove that $(D,V) \models \phi$ if and only if $(D,V) \models 
\phi^\tmono$.

\bigskip
\noindent \fbox{$\Rightarrow$} This direction is trivial.

\bigskip
\noindent \fbox{$\Leftarrow$} Assume $(D,V) \models \phi^\tmono$ and let 
$\Sigma$ be such that $(D,V) \models \mondbnfofo{\Sigma}{B}$. 
If $D = \nada$, then $\Sigma = \nada$ and $\mondbnfofo{\Sigma}{B}= 
\dbnfofo{\Sigma}$. 
Hence, assume $D \neq \nada$, and clearly $\Sigma \neq \nada$.

Because of the existential part of $\mondbnfofo{\Sigma}{B}$, every type $S \in
\Sigma$ has a `$B$-witness' in $\osmodel$, that is, an element $d_{S} \in D$
such that $(D,V) \models \tau^{B}_{S}(d_{S})$.
It is in fact safe to assume that all these witnesses are \emph{distinct}
(this is because $(D,V)$ can be proved to be $\ofo$-equivalent to such a model,
cf.~Proposition~\ref{p-1P}).
But because of the universal part of $\mondbnfofo{\Sigma}{B}$, we may assume
that for all states $d$ in $D$ there is a type $S_{d}$ in $\Sigma$ such that
$(D,V) \models \tau^{B}_{S_{d}}(d)$.
Putting these observations together we may assume that the map $d \mapsto S_{d}$
is surjective.

Note however, that where we have $(D,V) \models \tau^{B}_{S}(d)$, this does not 
necessarily imply that $(D,V) \models \tau_{S}(d)$: it might well be the case
that $d \in V(b)$ but $b \not\in S_{d}$, for some $ b \in B$.
What we want to do now is to shrink $V$ in such a way that the witnessed type
($S_d$) and the actually satisfied type coincide.
That is, we consider the valuation $U$ defined as $U^{\flat}(d) \isdef 
S_d$.\footnote{%
   Recall that a valuation $U:A\to\wp(D)$ can also be represented as a 
   colouring $U^{\flat}: D\to\wp(A)$ given by $U^{\flat}(d) \isdef 
   \{a \in A \mid d\in V(a)\}$.
   } 
It is then immediate by the surjectivity of the map $d \mapsto S_{d}$ that 
$(D,U) \models \dbnfofo{\Sigma}$, which implies that $(D,U) \models \phi$.

We now claim that 
\begin{equation}
\label{eq:moninf1}
U \leq_{B} V.
\end{equation}
To see this, observe that for $a \in A \setminus B$ we have the following 
equivalences:
\[
d \in U(a) \iff a \in S_{d} \iff (D,V) \models a(d) \iff d \in V(a),
\]
while for $b \in B$ we can prove
\[
d \in U(b) \iff b \in S_{d} \Longrightarrow (D,V) \models b(d) \iff d \in V(b).
\]
This suffices to prove \eqref{eq:moninf1}.

But from \eqref{eq:moninf1} and the earlier observation that $(D,U) \models 
\phi$ it is immediate by the monotonicity of $\phi$ in $B$ that $(D,V) \models 
\phi$.
\end{proof}

A careful analysis of the translation gives us the following 
corollary, providing normal forms for the monotone fragment of $\ofo$.

\begin{corollary}\label{cor:ofopositivenf}
For any sentence $\phi \in \ofo(A)$, the following hold.
\begin{enumerate}
\item 
The formula $\phi$ is monotone in $B \subseteq A$ iff it is equivalent to a 
formula in the basic form $\bigvee \mondbnfofo{\Sigma}{B}$ for some types 
$\Sigma \subseteq \wp(A)$.
\item The formula $\phi$ is monotone in every $a\in A$ 
iff $\phi$ is equivalent to a formula $\bigvee \posdbnfofo{\Sigma}$ for some
types $\Sigma \subseteq \wp(A)$.
\end{enumerate}
In both cases the norma forms are effective.
\end{corollary}

\subsection{Monotone fragment of $\ofoe$}

In order to prove the $\ofoe$-variant of Proposition~\ref{p:efftrans}, we need
to introduce some new notation. 

\begin{definition}
Let $B\subseteq A$ be a finite set of names. 
The $B$-monotone variant of $\dbnfofoe{\vlist{T}}{\Pi}$ is given as follows:
\begin{align*}
\mondbnfofoe{\vlist{T}}{\Pi}{B} 
  &\isdef  \exists \vlist{x}.\big(\arediff{\vlist{x}} 
     \land \bigwedge_i \tau^{B}_{T_i}(x_i) 
     \land \forall z.(\arediff{\vlist{x},z} \to \bigvee_{S\in \Pi} \tau^{B}_S(z))
     \big). 	
\end{align*}
When the set $B$ is a singleton $\{a\}$ we will write $a$ instead of $B$. 
The positive variant $\posdbnfofoe{\vlist{T}}{\Pi}$ of 
$\dbnfofoe{\vlist{T}}{\Pi}$ is defined as above but with $+$ in place of $B$.
\end{definition}

\begin{proposition}
\label{p:monofoeismonot}
There exists an effective translation $(-)^\tmono:\ofoe(A) \to \monot{\ofoe(A)}{B}$ such
that a sentence ${\phi \in \ofoe(A)}$ is monotone in $B$ if and only if 
$\phi\equiv \phi^\tmono$.
\end{proposition}

\begin{proof}
In proposition~\ref{p:mono-ofoei} this result is proved for $\ofoei$ (i.e.,
$\ofoe$ extended with generalised quantifiers). 
It is not difficult to adapt the proof for $\ofoe$. 
The translation is defined as follows. By Theorem \ref{thm:bnfofoe} ,without loss of generality, assume that $\phi$ is in basic normal form $\bigvee \dbnfofoe{\vlist{T}}{\Pi}$. Then 
 $\phi^\tmono \isdef \bigvee \mondbnfofoe{\vlist{T}}{\Pi}{B}$.
\end{proof}

Combining the normal form for $\ofoe$ and the proof of the above proposition, we 
therefore obtain a normal form for the monotone fragment of
$\ofoe$.

\begin{corollary}\label{cor:ofoepositivenf}
For any sentence $\phi \in \ofo(A)$, the following hold.
\begin{enumerate}
\item The formula $\phi$ is monotone in $B\subseteq A$ iff it is equivalent
to a formula in the basic form $\bigvee \mondbnfofoe{\vlist{T}}{\Pi}{B}$ where
for each disjunct we have $\vlist{T} \in \wp(A)^k$ for some $k$ and 
$\Pi\subseteq\vlist{T}$.
		
\item 
The formula $\phi$ is monotone in all $a\in A$ iff it is equivalent to a formula
in the basic form 
$\bigvee \posdbnfofoe{\vlist{T}}{\Pi}$ where for each disjunct we have 
$\vlist{T} \in \wp(A)^k$ for some $k$ and $\Pi\subseteq\vlist{T}$.
\end{enumerate}
In both cases, normal forms are effective.
\end{corollary}

\subsection{Monotone fragment of $\ofoei$}

First, in this case too we introduce some notation for the positive variant of
a sentence in normal form.

\begin{definition}
Let $B\subseteq A$ be a finite set of names. 
The $B$-positive variant of $\dbnfofoei{\vlist{T}}{\Pi}{\Sigma}$ is given as 
follows:
\begin{align*}
    \mondbnfofoei{\vlist{T}}{\Pi}{\Sigma}{B} 
  & \isdef  \mondbnfofoe{\vlist{T}}{\Pi \cup \Sigma}{B} 
       \land \mondbnfinf{\Sigma}{B}
\\  \mondbnfofoe{\vlist{T}}{\Lambda}{B} 
  & \isdef  \exists \vlist{x}.\big(\arediff{\vlist{x}} 
        \land \bigwedge_i \tau^{B}_{T_i}(x_i) 
	\land \forall z.(\arediff{\vlist{x},z} 
	    \to \bigvee_{S\in\Lambda} \tau^{B}_S(z))
	\big)
\\ \mondbnfinf{\Sigma}{B} 
  &\isdef  \bigwedge_{S\in\Sigma} \qu y.\tau^{B}_S(y) 
       \land \dqu y.\bigvee_{S\in\Sigma} \tau^{B}_S(y).
\end{align*}
When the set $B$ is a singleton $\{a\}$ we will write $a$ instead of $B$. 
The positive variant of $\dbnfofoei{\vlist{T}}{\Pi}{\Sigma}$ is defined as 
$\posdbnfofoei{\vlist{T}}{\Pi}{\Sigma} \isdef  
\mondbnfofoei{\vlist{T}}{\Pi}{\Sigma}{A}$.
\end{definition}

\noindent
We are now ready to proceed with the proof of the $\ofoei$-variant of
Proposition~\ref{p:efftrans} and thus to give the translation.

\begin{proposition}
\label{p:mono-ofoei}
There is an effective translation $(-)^\tmono:\ofoei(A) \to \monot{\ofoei(A)}{B}$
such that a sentence ${\phi \in \ofoei(A)}$ is monotone in $B$ if and only if
$\phi\equiv \phi^\tmono$.
\end{proposition}
\begin{proof}
By Theorem \ref{thm:bfofoei}, we assume that $\phi$ is in the normal form
$\bigvee\dbnfofoei{\vlist{T}}{\Pi}{\Sigma} = 
\dbnfofoe{\vlist{T}}{\Pi \cup \Sigma} \land \dbnfinf{\Sigma}$
for some sets of types $\Pi,\Sigma \subseteq \wp(A)$ and each $T_i \subseteq A$.
For the translation we define
\[
\Big(\bigvee \dbnfofoei{\vlist{T}}{\Pi}{\Sigma}\Big)^\tmono\isdef  
\bigvee \mondbnfofoei{\vlist{T}}{\Pi}{\Sigma}{B}.
\]
From the construction it is clear that $\phi^\tmono \in \monot{\ofoei(A)}{B}$
and therefore the right-to-left direction of the proposition is immediate by 
Proposition~\ref{p:monoismonot}. 
For the left-to-right direction assume that $\phi$ is monotone in $B$, we 
have to prove that $(D,V) \models \phi$ if and only if $(D,V) \models 
\phi^\tmono$.

\bigskip
\noindent \fbox{$\Rightarrow$} This direction is trivial.

\bigskip
\noindent \fbox{$\Leftarrow$}
Assume $(D,V) \models \phi^\tmono$, and in particular that $(D,V) \models 
\mondbnfofoei{\vlist{T}}{\Pi}{\Sigma}{B}$. 
If $D = \nada$, then $\Sigma = \Pi = \vlist{T} = \nada$ and 
$\mondbnfofoei{\vlist{T}}{\Pi}{\Sigma}{B}= \dbnfofoei{\vlist{T}}{\Pi}{\Sigma}$. 
Hence, assume $D \neq \nada$.
Observe that the elements of $D$ can be partitioned in the following way:
\begin{enumerate}[(a)]
\itemsep 0 pt
\item distinct elements $t_i \in D$ such that each $t_i$ satisfies 
   $\tau^{B}_{T_i}(x)$,
\item\label{it:dpi} for every $S \in \Sigma$ an infinite set $D_S$,
   such that every $d \in D_S$ satisfies $\tau^{B}_{S}$,
\item\label{it:ds} 
a \emph{finite} set $D_\Pi$ of elements, each satisfying one of the $B$-positive
   types $\tau^{B}_{S}$ with $S \in \Pi$.
\end{enumerate}
Following this partition, with every element $d\in D$ we may associate a type
$S_{d}$ in, respectively, (a)~$\vlist{T}$, (b)~$\Sigma$, or (c)~$\Pi$, such 
that $d$ satisfies $\tau^{B}_{S_{d}}$.
As in the proof of proposition~\ref{p:fomon}, we now consider the valuation $U$ 
defined as $U^{\flat}(d) \isdef S_d$, and as before we can show that
$U \leq_{B} V$.
Finally, it  easily from the definitions that $(D,U) \models 
\dbnfofoei{\vlist{T}}{\Pi}{\Sigma}$, implying that $(D,U) \models \phi$.
But then by the assumed $B$-monotonicity of $\phi$ it is immediate that $(D,V) 
\models \phi$, as required.
\end{proof}

As with the previous two cases, the translation provides normal forms for the monotone fragment of $\ofoei$.

\begin{corollary}\label{cor:ofoeipositivenf}
For any sentence $\phi \in \ofoei(A)$, the following hold:
\begin{enumerate}
\item 
The formula $\phi$ is monotone in $B \subseteq A$ iff it is equivalent to
a formula $\bigvee \mondbnfofoei{\vlist{T}}{\Pi}{\Sigma}{B}$ for 
$\Sigma\subseteq\Pi \subseteq \wp(A)$ and $\vlist{T} \in \wp(A)^k$ for some $k$.
\item 
The formula $\phi$ is monotone in every $a\in A$ iff it is equivalent to a
formula in the basic form $\bigvee \posdbnfofoei{\vlist{T}}{\Pi}{\Sigma}$ for 
types $\Sigma\subseteq \Pi \subseteq \wp(A)$ and $\vlist{T} \in \wp(A)^k$ for 
some $k$.
\end{enumerate}
In both cases, normal forms are effective.
\end{corollary}

\begin{proof}
We only remark that to obtain $\Sigma\subseteq\Pi$ in the above normal forms 
it is enough to use Proposition~\ref{prop:bfofoei-sigmapi} before applying 
the translation.
\end{proof}



\section{Continuity}
\label{sec-cont}

In this section we study the sentences that are \emph{continuous} in some set 
$B$ of monadic predicate symbols.

\begin{definition}\label{def:cont}
Let $U$ and $V$ be two $A$-valuations on the same domain $D$.
For a set $B \subseteq A$, we write $U \leq^{\omega}_{B} V$ if $U \leq_{B} V$ 
and $U(b)$ is finite, for every $b \in B$.

Given a monadic logic $\llang$ and a formula $\phi \in \llang(A)$ we say that
$\phi$ is \emph{continuous in $B \subseteq A$} if $\phi$ is monotone in $B$
and satisfies the following:
\begin{equation}
\label{eq:cont}
\text{if } (D, V), g \models \phi \text{ then } 
(D,  U), g \models \phi \text{ for some } U \leq^{\omega}_{B} V.
\end{equation}
for every monadic model $(D,V)$ and every assignment $g:\fovar\to D$.
\end{definition}

\begin{remark}\label{rem:contprodeach}
As for monotonicity, but with slightly more effort, one may show that a formula 
$\phi$ is continuous in a set $B$ iff it is continuous in every $b \in B$.
\end{remark}

What explains both the name and the importance of this property is its 
equivalence to so called \emph{Scott continuity}. 
To understand it, we may formalise the dependence of the meaning of a  monadic 
sentence $\phi$ with $m$-free variables $\vlist{x}$ in a one-step model $\osmodel
= (D,V )$ on a fixed name $b \in A$ as a map $\phi^\osmodel_b : \wp(D) \to 
\wp(D^m)$ defined by 
\[
X \subseteq D \mapsto \{ \vlist{d} \in D^m \mid  
(D,V[b \mapsto X]) \models \phi(\vlist{d}) \}.
\]
One can then verify that a sentence $\phi$ is continuous in $b$ if and only if 
the operation  $\phi^\osmodel_b$ is continuous with respect to the Scott 
topology on the powerset algebras\footnote{%
   The interested reader is referred to \cite[Sec. 8]{FV12} for a more precise 
   discussion of the connection.}. 
Scott continuity is of key importance in many areas of theoretical computer 
sciences where ordered structures play a role, such as domain theory (see 
e.g. \cite{abramsky1994domain}).

Similarly as for monotonicity, the semantic property of continuity can also be 
provided with a corresponding syntactical formulation.

\begin{definition}
Let $\llang \in \{ \ofo, \ofoe \}$
The fragment of $\ofo(A)$ of formulas that are 
\emph{syntactically continuous} in a subset $B \subseteq A$ is 
defined by the following grammar:
\[
\phi \defbnf \psi 
   \mid b(x) 
   \mid \phi \land \phi \mid \phi \lor \phi
   \mid \exists x.\phi,
\]
where $b\in B$ and $\psi \in \llang(A\setminus B)$. 
In both cases, we let $\cont{\llang(A)}{B}$ denote the set of $B$-continuous 
sentences.
\end{definition}

To define the syntactically continuous fragment of $\ofoei$, we first introduce
the following binary generalised quantifier $\wqu$: given two formulas $\phi(x)$ 
and $\psi$, we set
\[
\wqu x.(\phi,\psi) \isdef \forall x.(\phi(x) \lor \psi(x)) \land \dqu x.\psi(x).
\]
The intuition behind $\wqu$ is the following. If $(D,V),g \models \wqu x.(\phi,
\psi)$, then because of the second conjunct there are only finitely many $d \in
D$ refuting $\psi$. 
The point is that this weakens the universal quantification of the first conjunct
to the effect that only the finitely many mentioned elements refuting $\psi$ need
to satisfy $\phi$.

\begin{definition}
The fragment of $\ofoei(A)$-formulas that are 
\emph{syntactically continuous} in a subset $B \subseteq A$ is given by the 
following grammar:
\[
\phi \defbnf \psi 
   \mid a(x) 
   \mid \phi \land \phi \mid \phi \lor \phi 
   \mid \exists x.\phi \mid \wqu x.(\phi,\psi),
\]
where $b\in B$ and $\psi \in \ofoei(A\setminus B)$. 
We let $\cont{\ofoei(A)}{B}$ denote the set of $B$-continuous $\ofoei$-sentences.
\end{definition}

For $\ofo$ and $\ofoe$, the equivalence between the semantical and syntactical 
properties of continuity was established by van Benthem 
in \cite{van1997dynamic}. 

\begin{proposition}\label{fact:vb}
Let $\phi$ be a sentence of the monadic logic $\llang(A)$, where $\llang \in 
\{ \ofo, \ofoe\}$.
Then $\phi$ is continuous in a set $B \subseteq A$ if and only if there is a 
equivalent sentence $\phi^{\tcont} \in \cont{\llang(A)}{B}$.
\end{proposition}

\begin{proof}
The direction from right to left is covered by Proposition \ref{p:coniscont} 
below, so we immediately turn to the completeness part of the statement.
The case of $\ofo$ being treated in Subsection \ref{subsec:conofo}, we only discuss  the statement for $\ofoe$.
Hence, let $\phi \in \ofoe(A)$ be continuous in 
$B$. 
For simplicity in the exposition, we assume $B=\{b\}$, the case of an arbitrary 
$B$ being easily generalisable from what follows. 
Let $\vlist{y} \isdef y_0 \dots y_{k-1}$ be a list of $k$ variables not
occurring in $\phi$. 
Consider the formula $\phi_k( \vlist{y} )$ obtained from $\phi$  by substituting 
each occurrence of an atomic formula of the form $b(x)$ with the formula 
$\bigvee_{\ell < k} x = y_\ell$.
Define $\Phi \isdef \{ \exists \vlist{y}.\phi_k( \vlist{y} ) \mid k \in \omega \} \cup \{ \phi_{\emodel}\}$, 
where $ \phi_{\emodel} \isdef \forall x. \bot$ if $\emodel \models \phi$ and 
$ \phi_{\emodel} \isdef \exists x. \bot$ otherwise. 
By construction $\Phi \subset \cont{\ofoe(A)}{B}$.
Now, notice that $\lnot \Phi \cup \{\phi\}$ is inconsistent. 
Hence, by compactness of first-order logic, there is a $k \in \omega$ such that
$\phi \models \bigvee_{\ell < k} \exists \vlist{y}.\phi_k( \vlist{y} ) \lor 
  \phi_{\emodel}$. 
By monotonicity, $\exists \vlist{y}.\phi_k( \vlist{y} ) \models \phi$, for every
$k \in \omega$, and by definition $\phi_{\emodel} \models \phi$.
We therefore conclude that $\phi \equiv \bigvee_{\ell < k} \exists \vlist{y}.
\phi_k( \vlist{y} ) \lor \phi_\emodel$. 
As $\cont{\ofoe(A)}{B}$ is closed under disjunctions, this ends the proof of 
the statement.
\end{proof}

In this paper, we extend such a characterisation to $\ofoei$. 
Moreover, analogously to what we did in the previous section, for $\ofo$ and 
$\ofoei$ we provide both an explicit translation and a decidability result. 
From this latter perspective, the case of $\ofoe$ remains however open.

\begin{theorem}
\label{t:cont}
Let $\phi$ be a sentence of the monadic logic $\llang(A)$, where $\llang \in 
\{ \ofo, \ofoei \}$.
Then $\phi$ is continuous in a set $B \subseteq A$ if and only if there is a 
equivalent sentence $\phi^{\tcont} \in \cont{\llang(A)}{B}$.
Furthermore, it is decidable whether a sentence $\phi \in \llang(A)$ has this 
property or not.
\end{theorem}

Analogously to the previous case of monotonicity, the proof of the theorem is 
composed of two parts.  
We start with the right-left implication of the first claim (the preservation
statement), which also holds for $\ofoe$.

\begin{proposition}\label{p:coniscont}
Every sentence $\phi \in \cont{\llang(A)}{B}$ is continuous in $B$,
where $\llang\in\{ \ofo, \ofoe, \ofoei \}$.
\end{proposition}

\begin{proof}
First observe that $\phi$ is monotone in $B$ by 
Proposition~\ref{p:monoismonot}.
The case for $D=\nada$ being clear, we assume $D\neq\nada$.
We show, by induction, that any one-step formula $\phi$ in the fragment (which
may not be a sentence) satisfies \eqref{eq:cont}, for every non-empty one-step
model $(D, V)$ and assignment ${ g:\fovar\to D}$.
\begin{enumerate}[\textbullet]
\item 
If $\phi = \psi \in \llang(A\setminus B)$, changes in the $B$ part of the 
valuation will not affect the truth value of $\phi$ and hence the condition is 
trivial. 

\item
Case $\phi = b(x)$ for some $b \in B$: if $(D, V), g \models b(x)$ then 
$g(x)\in  V(b)$. 
Let $U$ be the valuation given by $U(b) \isdef \{ g(x) \}$,
$U(a) \isdef \nada$ for $a \in B \setminus \{b \}$ and
$U(a) \isdef V(a)$ for $a \in A \setminus B$.
Then it is obvious that $(D,  U), g \models b(x)$, while it is immediate by the
definitions that $U \leq^{\omega}_{B} V$.

\item 
Case $\phi = \phi_1 \lor \phi_2$: assume $(D, V), g \models \phi$. 
Without loss of generality we can assume that $(D, V), g \models \phi_1$ and 
hence by induction hypothesis there is $U \leq^{\omega}_{B} V$ such that $(D, U),
g \models \phi_1$ which clearly implies $(D, U), g \models 
\phi$. 

\item 
Case $\phi = \phi_1 \land \phi_2$: assume $(D, V), g \models \phi$. 
By induction hypothesis we have $U_1,U_2 \leq^{\omega}_{B} V$ such that 
$(D,U_{1}), g \models \phi_1$ and $(D, U_2), g \models \phi_2$.
Let $U$ be the valuation defined by putting $U(a) \isdef U_{1}(a) \cup U_{2}$; 
then clearly we have $U \leq^{\omega}_{B} V$, while it follows by monotonicity
that $(D,U), g \models \phi_1$ and $(D, U), g \models \phi_2$.
Clearly then $(D, U), g \models \phi$.

\item
Case $\phi = \exists x.\phi'(x)$ and $(D, V), g \models \phi$. 
By definition there exists $d\in D$ such that $(D, V), g[x\mapsto d] \models 
\phi'(x)$. 
By induction hypothesis there is a valuation $U \leq^{\omega}_{B} V$ such that 
$(D, U), g[x\mapsto d] \models \phi'(x)$ and hence $(D, U), g \models 
\exists x.\phi'(x)$.
\item
Case $\phi = \wqu x.(\phi',\psi)\in \cont{\ofoei(A)}{B}$ and $(D, V), g \models \phi$.  Define the formulas $\alpha(x)$ and 
$\beta$ as follows:
\[
\phi = 
\forall x.\underbrace{(\phi'(x) \lor \psi(x))}_{\alpha(x)} \land 
\underbrace{\dqu x.\psi(x)}_\beta.
\]

Suppose that $(D, V), g \models \phi$. 
By the induction hypothesis, for every $d \in D$ which  satisfies $(D, V), g_d
\models \alpha(x)$ (where we write $ g_d \isdef g[x\mapsto d]$) there is a 
valuation $U_d \leq^{\omega}_{B} V$ such that $(D, U_d), g_d \models \alpha(x)$. 
The crucial observation is that because of $\beta$, only finitely many elements
of $d$ refute $\psi(x)$. 
Let $U$ be the valuation defined by putting $U(a) \isdef \bigcup \{U_d(a) \mid 
(D, V), g_d \not\models \psi(x) \}$. 
Note that for each $b \in B$, the set $U(b)$ is a finite union of finite sets, 
and hence finite itself; it follows that $U \leq^{\omega}_{B} V$.
We claim that
\begin{equation}
\label{eq:1801}
(D, U), g \models \phi.
\end{equation}
It is clear that $(D, U), g \models \beta$ because $\psi$ (and hence $\beta$) 
is $B$-free.
To prove that $(D, U), g \models \forall x\, \alpha(x)$, take an arbitrary $d 
\in D$, then we have to show that $(D, U), g_d \models \phi'(x) \lor \psi(x)$. 
We consider two cases: If $(D, V), g_d \models \psi(x)$ we are done, again
because $\psi$ is $B$-free.
On the other hand, if $(D, V), g_d \not\models \psi(x)$, then $(D, U_d), g_d 
\models \alpha(x)$ by assumption on $U_{d}$, while it is obvious that $U_{d} 
\leq_{B} U$; but then it follows by monotonicity of $\alpha$ that $(D, U), g_d
\models \alpha(x)$.

\end{enumerate}
This finishes the proof.
\end{proof}

The second part of the proof of the theorem, is thus constituted by the following stronger version of the expressive completeness result that provides as a corollary normal forms for the syntactically continuous fragments.
\begin{proposition}
\label{p:efftranscont}
Let $\llang$ be one of the logics $\{ \ofo, \ofoei \}$.
There exists an effective translation $(-)^\tcont:\llang(A) \to \cont{\llang(A)}{B}$ such that
a sentence ${\phi \in \llang(A)}$ is continuous in $B \subseteq A$ if and only if 
$\phi\equiv \phi^\tcont$.
\end{proposition}

We prove the two manifestations of Proposition \ref{p:efftranscont} separately,
in two respective subsections.

By putting together the two propositions above, we are thence able to conclude.
\begin{proofof}{Theorem \ref{t:cont}}
The first claim follows from  Proposition~\ref{p:efftranscont}. 
Hence, by applying Fact \ref{f:decido} to Proposition \ref{p:efftranscont}, the 
problem of checking whether a sentence $\phi \in \llang(A)$ is continuous in 
$B \subseteq A$ or not, is decidable.
\end{proofof}

We conjecture that Proposition \ref{p:efftranscont}, and therefore
Theorem \ref{t:cont}, holds also for $\llang=\ofoe$.

\subsection{Continuous fragment of $\ofo$}\label{subsec:conofo}

Since continuity implies monotonicity, by Theorem \ref{t:mono}, in order to 
verify the $\ofo$-variant of Proposition \ref{p:efftranscont}, it is enough to
proof the following result.

\begin{proposition}\label{prop:ofocont}
There is an effective translation $(-)^\tcont:\monot{\ofo(A)}{B} \to 
\cont{\ofo(A)}{B}$ such that a sentence $\phi \in \monot{\ofo(A)}{B}$ is
continuous in $B \subseteq A$ if and only if $\phi\equiv \phi^\tcont$.
\end{proposition}

\begin{proof}
By Corollary \ref{cor:ofopositivenf}, to define the translation we assume, without loss of generality, that $\phi$
is in the basic form $\bigvee \mondbnfofo{\Sigma}{B}$.
For the translation, let
\[
(\bigvee \mondbnfofo{\Sigma}{B})^\tcont \isdef 
\bigvee \mondgbnfofo{\Sigma}{\Sigma^{-}_{B}}{B}
\]
where $\Sigma^{-}_{B} \isdef \{S\in \Sigma \mid B \cap S = \nada \}$.
From the construction it is clear that $\phi^\tcont \in \cont{\ofo(A)}{B}$ and
therefore the right-to-left direction of the proposition is immediate by 
Proposition~\ref{p:coniscont}. 

For the left-to-right direction assume that $\phi$ is continuous in $B$, we have
to prove that $(D, V) \models \phi$ iff $(D, V) \models \phi^\tcont$, for every 
one-step model $(D, V)$.
Our proof strategy consists of proving the same equivalence for the model 
$(D\times \omega, V_\pi)$, where $D\times\omega$ consists of $\omega$ many 
copies of each element in $D$ and $V_\pi$ is the valuation given by $V_{\pi}(a) 
\isdef \{(d,k) \mid d\in  V(a), k\in\omega\}$.
It is easy to see that $(D, V) \equiv^{\ofo} (D\times \omega, V_\pi)$ (see 
Proposition~\ref{p-1P}) and so it suffices indeed to prove that
\[
(D\times\omega, V_\pi) \models \phi
\text{ iff }
(D\times\omega, V_\pi) \models \phi^\tcont.
\]
Consider first $D= \nada$.
Then $(D\times\omega, V_\pi) = \emodel$, and therefore the claim is true since 
$\mondbnfofo{\nada}{B} = \mondgbnfofo{\nada}{\nada^{-}_{B}}{B}$ and
$\emodel \models  \mondbnfofo{\Sigma}{B}$ iff $\Sigma=\nada$. 
Hence, assume $D \neq \nada$.
\bigskip

\noindent \fbox{$\Rightarrow$}
Let $(D\times\omega, V_\pi) \models \phi$.
As $\phi$ is continuous in $B$ there is a valuation $U \leq^{\omega}_{B} V_\pi$ 
satisfying $(D\times\omega, U) \models \phi$. 
This means that  $(D\times\omega, U) \models \mondbnfofo{\Sigma}{B}$ for some
disjunct $\mondbnfofo{\Sigma}{B}$ of $\phi$.
Below we will use the following fact (which can easily be verified):
\begin{equation}
\label{eq:con101}
(D\times\omega),U \models \tau^{B}_{S}(d,k) \text{ iff }
S \setminus B = U^{\flat}(d,k) \setminus B \text{ and }
S \cap B \subseteq U^{\flat}(d,k).
\end{equation}

Our claim is now that $(D\times\omega, U) \models 
\mondgbnfofo{\Sigma}{\Sigma^{-}_{B}}{B}$. 

The existential part of $\mondgbnfofo{\Sigma}{\Sigma^{-}_{B}}{B}$ is trivially 
true. 
To cover the universal part, it remains to show that every element of 
$(D\times\omega, U)$ realizes a $B$-positive type in $\Sigma^{-}_{B}$.
Take an arbitrary pair $(d,k) \in D\times\omega$ and let $T$ be the (full) type 
of $(d,k)$, that is, let $T \isdef U^{\flat}(d,k)$.
If $B \cap T = \nada$ then trivially $T\in \Sigma^{-}_{B}$ and we are done. 
So suppose $B \cap T \neq \nada$. 
Observe that in $D\times\omega$ we have infinitely many copies of $d\in D$.
Hence, as $U(b)$ is finite for every $b \in B$, there must be some $(d,k')$ 
with type $U^{\flat}(d,k') = V_{\pi}^{\flat}(d,k') \setminus B = 
V_{\pi}^{\flat}(d,k) \setminus B = T \setminus B$.
It follows from $(D\times\omega, U) \models \mondbnfofo{\Sigma}{B}$ and 
\eqref{eq:con101} that there is some $S \in \Sigma$ such that
$S \setminus B = U^{\flat}(d,k') \setminus B = U^{\flat}(d,k')$ and 
$S \cap B \subseteq U^{\flat}(d,k) \cap B = \nada$.
From this we easily derive that $S = U^{\flat}(d,k')$ and $S \in \Sigma^{-}_{B}$.
Finally, we observe that $S \setminus B = U^{\flat}(d,k') \setminus B =
U^{\flat}(d,k) \setminus B$ and $S \cap B = \nada \subseteq U^{\flat}(d,k)$, 
so that by \eqref{eq:con101} we find that $D \times \omega,U) \models 
\tau^{B}_{S}(d,k)$ indeed.

Finally, by monotonicity it directly follows from
$(D\times\omega, U) \models \mondgbnfofo{\Sigma}{\Sigma^{-}_{B}}{B}$
that 
$(D\times\omega, V_{\pi}) \models \mondgbnfofo{\Sigma}{\Sigma^{-}_{B}}{B}$,
and from this it is immediate that $(D\times\omega, V_\pi) \models \phi^\tcont$.
\bigskip

\noindent \fbox{$\Leftarrow$}
Let $(D\times\omega, V_\pi) \models 
\mondgbnfofo{\Sigma}{\Sigma^{-}_{B}}{B}$.
To show that $(D\times\omega, V_\pi) \models \mondbnfofo{\Sigma}{B}$, the 
existential part is trivial. For the universal part just observe that 
$\Sigma^{-}_{B} \subseteq \Sigma$.
\end{proof}

A careful analysis of the translation gives us the following corollary,
providing normal forms for the continuous fragment of $\ofo$.

\begin{corollary}\label{cor:ofocontinuousnf}
For any sentence $\phi \in \ofo(A)$, the following hold.
\begin{enumerate}
\item 
The formula $\phi$ is continuous in $B \subseteq A$ iff it is equivalent to a 
formula $\bigvee \mondgbnfofo{\Sigma}{\Sigma^{-}_{B}}{B}$ for some types $\Sigma
\subseteq \wp(A)$, where $\Sigma^{-}_{B} \isdef \{S\in \Sigma \mid B \cap S = 
\nada \}$.
\item 
If $\phi$ is monotone in $A$ 
then $\phi$ is continuous in $B \subseteq A$ iff it is equivalent to a formula 
in the basic form $\bigvee \posdgbnfofo{\Sigma}{\Sigma^{-}_{B}}$ for some types 
$\Sigma \subseteq \wp(A)$, where $\Sigma^{-}_{B} \isdef \{S\in \Sigma \mid B \cap
S = \nada \}$.
\end{enumerate}
\end{corollary}


\subsection{Continuous fragment of $\ofoei$}

As for the previous case, the $\ofoei$-variant of
Proposition \ref{p:efftranscont} is an immediate consequence of 
Theorem \ref{t:mono} and the following proposition.

\begin{proposition}\label{lem:ofoeictrans}
There is an effective translation $(-)^\tcont:\monot{\ofoei(A)}{B} \to 
\cont{\ofoei(A)}{B}$ such that a sentence $\phi \in \monot{\ofoei(A)}{B}$ is 
continuous in $B$ if and only if $\phi\equiv \phi^\tcont$.
\end{proposition}

\begin{proof}
By Corollary \ref{cor:ofoeipositivenf}, we assume that $\phi$ is in basic normal form, i.e., $\phi = \bigvee 
\mondbnfofoei{\vlist{T}}{\Pi}{\Sigma}{B}$.
For the translation let 
$\big(\bigvee \mondbnfofoei{\vlist{T}}{\Pi}{\Sigma}{B}\big)^\tcont \isdef 
\bigvee \mondbnfofoei{\vlist{T}}{\Pi}{\Sigma}{B}^\tcont$ where
\[
\mondbnfofoei{\vlist{T}}{\Pi}{\Sigma}{B}^\tcont \isdef
\begin{cases}
   \bot &\text{ if } B \cap \bigcup \Sigma \neq \nada
\\ \mondbnfofoei{\vlist{T}}{\Pi}{\Sigma}{B} &\text{ otherwise}.
\end{cases}
\]

First we prove the right-to-left direction of the proposition. 
By Proposition~\ref{p:coniscont} it is enough to show that $\phi^\tcont \in
\cont{\ofoei(A)}{B}$. 
We focus on the disjuncts of $\phi^\tcont$. 
The interesting case is where $B \cap \bigcup \Sigma = \nada$.
If we rearrange $\mondbnfofoei{\vlist{T}}{\Pi}{\Sigma}{B}$ somewhat and define
the formulas $\phi', \psi$ as follows:
\begin{align*}
\exists \vlist{x}.\Big(
   & \arediff{\vlist{x}} \land \bigwedge_i \tau^B_{T_i}(x_i)\ 
     \land \forall z.(\underbrace{\lnot\arediff{\vlist{x},z} 
         \lor \bigvee_{S\in \Pi} \tau^B_S(z)}_{\phi'(\vlist{x},z)} 
         \lor \underbrace{\bigvee_{S\in \Sigma} \tau^B_S(z)}_{\psi(z)})\ 
     \land \dqu y.\underbrace{\bigvee_{S\in\Sigma} \tau^B_S(y)}_{\psi(y)} \Big)
\\ &  
     \land \bigwedge_{S\in\Sigma} \qu y.\tau^B_S(y).
\end{align*}
Then we find that
\[
\mondbnfofoei{\vlist{T}}{\Pi}{\Sigma}{B} \equiv 
\exists \vlist{x}.\Big(\arediff{\vlist{x}} 
     \land \bigwedge_i \tau^B_{T_i}(x_i) \land \wqu z.(\phi'(\vlist{x},z),\psi(z)) 
      \Big) 
   \land \bigwedge_{S\in\Sigma} \qu y.\tau^B_S(y),
\]
which belongs to the required fragment because $B \cap \bigcup \Sigma = \nada$.

For the left-to-right direction of the proposition we have to prove that $\phi \equiv
\phi^\tcont$.

\bigskip
\noindent\fbox{$\Rightarrow$} 
Let $(D, V) \models \phi$. 
Because $\phi$ is continuous in $B$ we may assume that $ V(b)$ is finite, for
all $b \in B$.
Let $\mondbnfofoei{\vlist{T}}{\Pi}{\Sigma}{B}$ be a disjunct of $\phi$ such that
$(D, V) \models \mondbnfofoei{\vlist{T}}{\Pi}{\Sigma}{B}$.
If $D = \nada$, then $ {\vlist{T}}={\Pi}={\Sigma}=\nada$, and 
$\mondbnfofoei{\vlist{T}}{\Pi}{\Sigma}{B} =
(\mondbnfofoei{\vlist{T}}{\Pi}{\Sigma}{B})^\tcont$. 
Hence, let $D \neq \nada$.
Suppose for contradiction that $B \cap \bigcup \Sigma \neq \nada$, then there 
must be some $S\in\Sigma$ with $B \cap S \neq \nada$. 
Because $(D, V) \models \mondbnfofoei{\vlist{T}}{\Pi}{\Sigma}{B}$ we have, in
particular, that $(D, V) \models \qu y.\tau^B_S(x)$ and hence $V(b)$ must be 
infinite, for any $b \in B \cap S$, which is absurd.
It follows that $B \cap \bigcup \Sigma = \nada$, but then we trivially conclude
that $(D, V) \models \phi^\tcont$ because the disjunct remains unchanged. 

\bigskip
\noindent\fbox{$\Leftarrow$} 
Let $(D, V) \models \phi^\tcont$. 
The only difference between $\phi$ and $\phi^\tcont$ is that some disjuncts may 
have been replaced by $\bot$. Therefore this direction is trivial.
\end{proof}

We conclude the section by stating the following corollary, 
providing normal forms for the continuous fragment of $\ofoei$. 

\begin{corollary}\label{cor:ofoeicontinuousnf}
For any sentence $\phi \in \ofoei(A)$, the following hold.
\begin{enumerate}
\item \label{pt:ofoeifcontinuous}
The formula $\phi$ is continuous in $B \subseteq A$ iff 
$\phi$ is equivalent to a 
formula, effectively obtainable from $\phi$, which is a disjunction of 
formulas $\mondbnfofoei{\vlist{T}}{\Pi}{\Sigma}{B}$
where $\vlist{T},\Sigma$ and $\Pi$ are such that $\Sigma \subseteq \Pi \subseteq 
\vlist{T}$ and $B \cap \bigcup\Sigma = \nada$. 	
\item \label{pt:ofoeimonotone}
If $\phi$ is monotone 
(i.e., $\phi\in{\ofoei}^+(A)$) then 
$\phi$ is continuous in $B \subseteq A$
iff
it is equivalent to a formula, effectively obtainable from $\phi$, which is a 
disjunction of formulas
$\bigvee \posdbnfofoei{\vlist{T}}{\Pi}{\Sigma}$,
where $\vlist{T},\Sigma$ and $\Pi$ are such that $\Sigma \subseteq \Pi \subseteq 
\vlist{T}$ and $B \cap \bigcup\Sigma = \nada$. 	
\end{enumerate}
\end{corollary}

\begin{proof}
Notice that, from Proposition \ref{prop:bfofoei-sigmapi}, every sentence in the
basic form $\bigvee \dbnfofoei{\vlist{T}}{\Pi}{\Sigma}$
can be assumed such that $\Sigma \subseteq 
\Pi \subseteq \vlist{T}$. The claims hence follow by construction of the translation.
\end{proof}



\section{Submodels and quotients}
\label{sec:inv}

There are various natural notions of morphism between monadic models; the one 
that we will be interested here is that  of a (strong) homomorphism.

\begin{definition}
\label{d:hom}
Let $\osmodel = (D,V)$ and $\osmodel' = (D',V')$ be two monadic models.
A map $f: D \to D'$ is a \emph{homomorphism} from $\osmodel$ to $\osmodel'$, 
notation: $f: \osmodel \to \osmodel'$, if we have $d \in V(a)$ iff $f(d) \in 
V'(a)$, for all $a \in A$ and $d \in D$.
\end{definition}

In this section we will be interested in the sentences of $\ofo, \ofoe$ and 
$\ofoei$ that are preserved under taking submodels and the ones that are 
invariant under quotients.

\begin{definition}
\label{d:inv}
Let $\osmodel = (D,V)$ and $\osmodel' = (D',V')$ be two monadic models.
We call $\osmodel$ a \emph{submodel} of $\osmodel'$ if $D \subseteq D'$ and 
the inclusion map $\iota_{DD'}: D \hookrightarrow D'$ is a homomorphism, and 
we say that $\osmodel'$ is a \emph{quotient} of $\osmodel$ if there 
is a surjective homomorphism $f: \osmodel \to \osmodel'$.

Now let $\phi$ be an $\llang$-sentence, where $\llang \in \{ \ofo, \ofoe, \ofoei
\}$.
We say that $\phi$ is \emph{preserved under taking submodels} if 
$\osmodel \models \phi$ implies $\osmodel' \models \phi$, whenever
$\osmodel'$ is a submodel of $\osmodel$.
Similarly, $\phi$ is \emph{invariant under taking quotients} if we have
$\osmodel \models \phi$ iff $\osmodel' \models \phi$, whenever $\osmodel'$ is
a quotient of $\osmodel$.
\end{definition}

The first of these properties (preservation under taking submodels) is well
known from classical model theory --- it is for instance the topic of the
{\L}os-Tarski Theorem.
When it comes to quotients, in model theory one is usually more interested in
the formulas that are \emph{preserved} under surjective homomorphisms (and
the definition of homomorphism may also differ from ours): for instance, this 
is the property that is characterised by Lyndon's Theorem.
Our preference for the notion of \emph{invariance} under quotients stems from 
the fact that the property of invariance under quotients plays a key role in
characterising the \emph{bisimulation-invariant fragments} of various monadic
second-order logics, as is explained in our companion paper~\cite{companionWEAK}.


\subsection{Preservation under submodels}

In this subsection we characterise the fragments of our predicate logics 
consisting of the sentences that are preserved under taking submodels.
That is, the main result of this subsection is a {\L}os-Tarksi Theorem for
$\ofoei$.

\begin{definition}
The \emph{universal fragment} of the set $\ofoei(A)$ is the collection 
$\univ{\ofoei(A)}$ of formulas given by the following grammar:
\[
\varphi \defbnf
\top \mid \bot 
\mid a(x)
\mid \neg a(x)
\mid x \foeq y
\mid x \foneq y
\mid \varphi \lor \varphi
\mid \varphi \land \varphi
\mid \forall x.\varphi
\mid \dqu x.\varphi
\]
where $x,y\in \fovar$ and $a \in A$.
The universal fragment $\univ{\ofoe(A)}$ is obtained by deleting the clause for
$\dqu$ from this grammar, and we obtain the universal fragment $\univ{\ofo(A)}$ 
by further deleting both clauses involving the equality symbol.
\end{definition}

\begin{theorem}
\label{t:univ}
Let $\phi$ be a sentence of the monadic logic $\llang(A)$, where $\llang \in 
\{ \ofo, \ofoe, \ofoei \}$.
Then $\phi$ is preserved under taking submodels if and only if there is a 
equivalent formula $\phi^{\tuniv} \in \univ{\llang(A)}$.
Furthermore, it is decidable whether a sentence $\phi \in \llang(A)$ has this 
property or not.
\end{theorem}

We start by verifying that universal formulas satisfy the property. 

\begin{proposition}
\label{p:univ1}
Let $\phi \in \univ{\llang(A)}$ be a universal sentence of the monadic logic $\llang(A)$, where $\llang \in 
\{ \ofo, \ofoe, \ofoei \}$. 
Then $\phi$ is preserved under taking submodels.
\end{proposition}

\begin{proof} 
It is enough to directly consider the case $\llang = \ofoei$.
Let $(D',V')$ be a submodel of the monadic model $(D,V)$. 
The case for $D = \nada$ being immediate, let us assume $D \neq \nada$.
By induction on the complexity of a formula $\phi \in \univ{\ofoei(A)}$ we will 
show that for any assignment $g: \fovar \to D'$ we have
\[
(D,V),g' \models \phi \text{ implies } (D',V'),g \models \phi,
\]
where $g':= g \circ \iota_{D'D}$.
We will only consider the inductive step of the proof where $\phi$ is of the 
form $\dqu x. \psi$.
Define $X_{D,V} \isdef \{ d \in D \mid (D,V), g'[x \mapsto d] \models \psi \}$,
and similarly, $X_{D',V'} \isdef \{ d \in D' \mid (D',V'), g[x \mapsto d] \models 
\psi \}$.
By the inductive hypothesis we have that $X_{D,V} \cap D' \subseteq X_{D',V'}$,
implying that $D' \setminus X_{D',V'} \subseteq D \setminus X_{D,V}$.
But from this we immediately obtain that 
\[
|D \setminus X_{D,V}| < \omega
\text{ implies } |D' \setminus X_{D',V'}| < \omega,
\]
which means that $(D,V),g' \models \phi$ implies $ (D',V'),g \models \phi$, as 
required. 
\end{proof}

Before verifying the  `hard' side of the theorem, we define the appropriate translations from each monadic logic into its universal fragment.

\begin{definition}
\label{d:univ_trans}
We start by defining the translations for sentences in basic normal forms. 

For  $\ofo$-sentences in basic form we first set
\[
          \Big(  \dbnfofo{\Sigma}  \Big)^{\tuniv} 
         \isdef 
   \forall z \bigvee_{S \in \Sigma} \tau_{S}(z)
\]
Second, we define $(\bigvee_{i} \alpha_{i})^{\tuniv} \isdef \bigvee 
\alpha_{i}^{\tuniv}$. 
Finally, we extend the translation $(-)^{\tuniv}$ to the collection of all 
$\ofo$-sentences by defining $\phi^{\tmoda} \isdef (\tbas{\phi})^{\tuniv}$, 
where $\tbas{\phi}$
is the basic normal form of $\phi$ as given by Fact~\ref{fact:ofonormalform}.

Similarly, for $\ofoe$-sentences we first define
\[
        \Big(  \dbnfofoe{\vlist{T}}{\Pi}  \Big)^{\tuniv} 
         \isdef 
   \forall z \bigvee_{S \in \vlist{T}\cup \Pi} \tau_{S}(z)
\]
and then we extend it to the full language by distributing over disjunction and applying Theorem~\ref{thm:bnfofoe} to convert an arbitrary $\ofoe$ sentence into an equivalent sentence in basic normal form.

Finally, for simple basic formulas of $\ofoei$, the translation $(-)^\tuniv$ is given as follows:
\[
(\dbnfofoei{\vlist{T}}{\Pi}{\Sigma})^{\tuniv} \isdef 
   \forall z \bigvee_{S \in \vlist{T}\cup \Pi\cup\Sigma} \tau_{S}(z) 
   \land \dqu z \bigvee_{S \in \Sigma} \tau_{S}(z).
\]
The definition is thus extended to the full language $\ofoei$ as expected: 
given a $\ofoei$-sentence $\phi$, by Theorem~\ref{thm:bfofoei} and
Proposition~\ref{prop:bfofoei-sigmapi} we compute an equivalent
basic form $\bigvee \dbnfofoei{\vlist{T}}{\Pi}{\Sigma}$, with 
$\Sigma \subseteq \Pi \subseteq \vlist{T}$, and therefore we set
$\phi^{\tuniv}\isdef \bigvee (\dbnfofoei{\vlist{T}}{\Pi}{\Sigma})^{\tuniv}$.
\end{definition}

The missing parts in the proof of the theorem is thence covered by the following
result.

\begin{proposition} 
\label{p:univ2}
For any monadic logic $\llang \in \{ \ofo, \ofoe, \ofoei \}$ there is an 
effective translation $(-)^\tuniv: \llang(A) \to \univ{\llang(A)}$ such that a
sentence $\phi \in \llang(A)$ is preserved under taking submodels if and only
if $\phi\equiv \phi^\tuniv$.
\end{proposition}

\begin{proof}
We only consider the case where $\llang = \ofoei$, leaving the other cases to 
the reader.

It is easy to see that $\phi^{\tuniv} \in \univ{\ofoei(A)}$, for every sentence
$\phi \in \ofoei(A)$; but then it is immediate by Proposition~\ref{p:univ1} that 
$\phi$ is preserved under taking submodels if $\phi \equiv \phi^{\tuniv}$.

For the left-to-right direction, assume that $\phi$ is preserved 
under taking submodels.
It is easy to see that $\phi$ implies $\phi^{\tuniv}$, so we focus on proving
the opposite.
That is, we suppose that $(D,V) \models \phi^{\tuniv}$, and aim to show that 
$(D,V) \models \phi$.
 
By Theorem \ref{thm:bfofoei} and Proposition~\ref{prop:bfofoei-sigmapi} we may assume without loss of 
generality that $\phi$ is a disjunction of setences of the form 
$\dbnfofoei{\vlist{T}}{\Pi}{\Sigma}$, where $\Sigma \subseteq \Pi \subseteq 
\vlist{T}$.
It follows that $(D,V)$ satisfies some  disjunct 
$   \forall z \bigvee_{S \in \vlist{T}\cup \Pi\cup\Sigma} \tau_{S}(z) 
   \land \dqu z \bigvee_{S \in \Sigma} \tau_{S}(z)$
of $\Big(\dbnfofoei{\vlist{T}}{\Pi}{\Sigma}\Big)^{\tuniv}$.
Expand $D$ with finitely many elements $\vlist{d}$, in one-one correspondence 
with $\vlist{T}$, and ensure that the type of each $d_{i}$ is $T_{i}$.
In addition, add, for each $S \in \Sigma$, infinitely many elements 
$\{ e^{S}_{n} \mid n \in \omega\}$, each of type $S$.
Call the resulting monadic model $\osmodel' = (D',V')$.

This construction is tailored to ensure that 
$(D',V') \models \dbnfofoei{\vlist{T}}{\Pi}{\Sigma}$, and so we obtain $(D',V') 
\models \phi$.
But obviously, $\osmodel$ is a submodel of $\osmodel'$, whence $(D,V) \models 
\phi$ by our assumption on $\phi$.
\end{proof}

\begin{proofof}{Theorem \ref{t:univ}}
The  first part of  the theorem is an immediate consequence of 
Proposition~\ref{p:univ2}. By applying Fact \ref{f:decido} to 
Proposition~\ref{p:univ2} we finally obtain that for the three concerned 
formalisms the problem of deciding whether a sentence is preserved under taking 
submodels is decidable.
\end{proofof}

As an immediate consequence of the proof of the previous 
Proposition~\ref{p:univ2}, we get effective normal forms for the universal
fragments.

\begin{corollary}\label{cor:univ}
The following hold:
\begin{enumerate}
\item\label{cor:ofo} 
A sentence $\phi \in \ofoe(A)$ is preserved under taking submodels iff it is 
equivalent to a formula 
$\bigvee \big( \forall z \bigvee_{S \in \Sigma} \tau_{S}(z)\big)$, for  types 
$\Sigma\subseteq \wp(A)$.
\item\label{cor:ofoe} 
A sentence $\phi \in \ofoe(A)$ is preserved under taking submodels iff it is 
equivalent to a formula 
$\bigvee \big( \forall z \bigvee_{S \in \vlist{T}\cup \Pi} \tau_{S}(z)\big)$, 
for types $\Pi \subseteq \wp(A)$ and $\vlist{T} \in \wp(A)^k$ for some $k$.
\item\label{cor:ofoei} 
A sentence $\phi \in \ofoei(A)$ is preserved under taking submodels iff it is 
equivalent to a formula 
$\bigvee \big(\forall z \bigvee_{S \in \vlist{T}\cup \Pi\cup\Sigma} \tau_{S}(z) 
   \land \dqu z \bigvee_{S \in \Sigma} \tau_{S}(z)\big)$ 
for types $\Sigma\subseteq \Pi \subseteq \wp(A)$ and $\vlist{T} \in \wp(A)^k$
for some $k$.
\end{enumerate}
In all three cases, normal forms are effective.
\end{corollary}


\subsection{Invariance under quotients}

The following theorem states that monadic first-order logic \emph{without}
equality ($\ofo$) provides the quotient-invariant fragment of both monadic 
first-order logic with equality ($\ofoe$), and of infinite-monadic predicate 
logic ($\ofoei$).

\begin{theorem}
\label{t:qinv}
Let $\phi$ be a sentence of the monadic logic $\llang(A)$, where $\llang \in 
\{ \ofoe, \ofoei \}$.
Then $\phi$ is invariant under taking quotients if and only if there is a 
equivalent sentence in $\ofo$.
Furthermore, it is decidable whether a sentence $\phi \in \llang(A)$ has this 
property or not.
\end{theorem}

We first state the `easy' part of the first claim of the theorem. 
Note that in fact, we have already been using this observation in earlier parts
of the paper.

\begin{proposition}
\label{p:m-qinv}
Every sentence in $\ofo$ is invariant under taking quotients.
\end{proposition}

\begin{proof}
Let $f: D \to D'$ provide a surjective homomorphism between the models $(D,V)$ 
and $(D',V')$, and 
observe that for any assignment $g: \fovar \to D$ on $D$, the composition $f 
\circ g: \fovar \to D'$ is an assignment on $D'$.

In order to prove the proposition one may show that, for an arbitrary 
$\ofo$-formula $\phi$ and an arbitrary assignment $g: \fovar \to D$, we have
\begin{equation}
\label{eq:m-qinv1}
(D,V),g \models \phi \text{ iff } (D',V'), f \circ g \models \phi.
\end{equation}
We leave the proof of \eqref{eq:m-qinv1}, which proceeds by a straightforward 
induction on the complexity of $\phi$, as an exercise to the reader.
\end{proof}

To prove the remaining part of Theorem~\ref{t:qinv}, we start with providing 
translations from respectively $\ofoe$ and $\ofoei$ to $\ofo$.

\begin{definition}
\label{d:7-21}
For $\ofoe$-sentences in basic form we first define
\[
   \Big( \dbnfofoe{\vlist{T}}{\Pi} \Big)^{\tmoda} 
\isdef \bigwedge_{i} \exists x_i. \tau_{T_i}(x_i) \land 
\forall x. \bigvee_{S\in\Pi} \tau_S(x),
\]
whereas for $\ofoei$-sentences in basic form we start with defining
\[
\Big( \dbnfofoei{\vlist{T}}{\Pi}{\Sigma} \Big)^{\tmodb} 
\isdef \bigwedge_{i} \exists x_i. \tau_{T_i}(x_i) \land 
\forall x. \bigvee_{S\in\Sigma} \tau_S(x).
\]
In both cases, the translations is then extended to the full language as in Definition \ref{d:univ_trans}.
\end{definition}

Note that the two translations may give \emph{different} translations for 
$\ofoe$-sentences.
Also observe that the $\Pi$ `disappears' in the translation of the formula
$\dbnfofoei{\vlist{T}}{\Pi}{\Sigma}$.

The key property of these translations is the following.

\begin{proposition}
\label{p-1P}
\begin{enumerate}
\item
For every one-step model $(D,V)$ and every $\phi \in \ofoe(A)$ we have
\begin{equation}
\label{eq-0cr}
(D,V) \models \phi^{\tmoda} \text{ iff } (D\times \omega,V_\pi) \models \phi.
\end{equation}
\item
For every one-step model $(D,V)$ and every $\phi \in \ofoei(A)$ we have
\begin{equation}
\label{eq-1cr}
(D,V) \models \phi^{\tmodb} \text{ iff } (D\times \omega,V_\pi) \models \phi.
\end{equation}
\end{enumerate}
\noindent
Here $V_{\pi}$ is the induced valuation given by 
$V_{\pi}(a) \isdef \{ (d,k) \mid d \in V(a), k\in\omega\}$.
\end{proposition}

\begin{proof}
We only prove the claim for $\ofoei$ (i.e., the second part of the proposition),
the case for $\ofoe$ being similar.
Clearly it suffices to prove \eqref{eq-1cr} for formulas of the form
$\alpha = \dbnfofoei{\vlist{T}}{\Pi}{\Sigma}$.

First of all, if $\osmodel$ is the empty model, we find ${\vlist{T}}={\Pi} = 
{\Sigma} = \nada$, $(D, V) = (D\times \omega, V_\pi)$, and
$\dbnfofoei{\vlist{T}}{\Pi}{\Sigma} = 
(\dbnfofoei{\vlist{T}}{\Pi}{\Sigma})^\tmodb$. 
In other words, in this case there is nothing to prove.
\smallskip

In the sequel we assume that $D \neq \nada$.

\noindent\fbox{$\Rightarrow$} 
Assume $(D, V) \models \phi^{\tmodb}$, we will show that 
$(D\times \omega, V_\pi) \models \dbnfofoei{\vlist{T}}{\Pi}{\Sigma}$.
Let $d_i$ be such that $V^{\flat}(d_i) = T_{i}$ in $(D, V)$. 
It is clear that the $(d_i,i)$ provide \emph{distinct} elements, with each 
$(d_i,i)$ satisfying $\tau_{T_i}$ in $(D\times\omega, V_{\pi})$ and therefore 
the first-order existential part of $\alpha$ is satisfied. 
With a similar argument it is straightforward to verify that the $\qu$-part of 
$\alpha$ is also satisfied --- here we critically use the observation that
$\Sigma \subseteq \vlist{T}$, so that every type in $\Sigma$ is witnessed in 
the model $(D,V)$, and hence witnessed infinitely many times in $(D\times\omega,
V_\pi)$.

For the universal parts of $\dbnfofoei{\vlist{T}}{\Pi}{\Sigma}$ it is enough
to observe that, because of the universal part of $\alpha^\tmodb$, \emph{every}
$d\in D$ realizes a type in $\Sigma$. 
By construction, the same applies to $(D\times\omega, V_{\pi})$, 
therefore this takes care of both universal quantifiers.
\medskip
		
\noindent\fbox{$\Leftarrow$} 
Assuming that $(D\times \omega, V_\pi) \models 
\dbnfofoei{\vlist{T}}{\Pi}{\Sigma}$,
we will show that $(D, V) \models \phi^\tmodb$. 
The existential part of $\alpha^{\tmodb}$ is trivial. 
For the universal part we have to show that every element of $D$ realizes a 
type in $\Sigma$. 
Suppose not, and let $d\in D$ be such that $\lnot\tau_S(d)$ for all $S\in 
\Sigma$. 
Then we have $(D\times\omega, V_\pi) \not\models \tau_S(d,k)$ for all $k$.
That is, there are infinitely many elements not realising any type in $\Sigma$. 
Hence we have $(D\times\omega, V_\pi) \not\models \dqu y.\bigvee_{S\in\Sigma} 
\tau_S(y)$. 
Absurd, because this formula is a conjunct of 
$\dbnfofoei{\vlist{T}}{\Pi}{\Sigma}$.
\end{proof}

\noindent
We will now show how the theorem follows from this.
First of all we verify that  in both cases $\ofo$ is expressively complete
for the property of being invariant under taking quotients.

\begin{proposition} 
\label{p:invq}
For any monadic logic $\llang \in \{ \ofoe, \ofoei \}$ there is an effective 
translation $(-)^\tinvq: \llang(A) \to \ofo$ such that a sentence $\phi \in
\llang(A)$ is invariant under taking quotients if and only if
$\phi\equiv \phi^\tinvq$.
\end{proposition}
\begin{proof}
Let $\phi$ be a sentence of $\ofoei$, and let $\phi^\tinvq \isdef \phi^{\tmodb}$ (we only cover the case of $\llang = 
\ofoei$, the case for $\llang = \ofoe$ is similar, just take $\phi^\tinvq \isdef \phi^{\tmoda}$)
We will show that 
\begin{equation}
\label{eq:m-qinv2}
\phi \equiv \phi^{\tinvq} \text{ iff } \text{$\phi$ is invariant under 
taking quotients}.
\end{equation}
The direction from right to left is immediate by Proposition~\ref{p:m-qinv}.
For the other direction it suffices to observe that any model $(D,V)$ is a 
quotient of its `$\omega$-product' $(D\times \omega, V_\pi)$, and to 
reason as follows:
\begin{align*}
(D,V) \models \phi 
   & \text{ iff } (D\times \omega, V_\pi) \models \phi
   & \text{(assumption on $\phi$)}
\\ & \text{ iff } (D,V) \models \phi^{\tmodb}
   & \text{(Proposition~\ref{p-1P})}
\end{align*}

\end{proof}

Hence we can conclude.
\begin{proofof}{Theorem~\ref{t:qinv}}
The theorem is an immediate consequence of Proposition \ref{p:invq}.
Finally, the effectiveness of translation  $(\cdot)^{\tmodb}$, decidability of $\ofoei$ (Fact \ref{f:decido}) and \eqref{eq:m-qinv2} yield
that it is decidable whether a given $\ofoei$-sentence $\phi$ is invariant under 
taking quotients or not.
\end{proofof}

As a corollary, we therefore obtain:

\begin{corollary}\label{cor:qinv}
Let $\phi$ be a sentence of the monadic logic $\llang(A)$, where $\llang \in 
\{ \ofoe, \ofoei \}$. Then $\phi$ is invariant under taking quotients if and only if there is a 
equivalent sentence $\bigvee \big(\bigwedge_{S\in\Sigma} \exists x. \tau_S(x) \land 
   \forall x. \bigvee_{S\in\Sigma} \tau_S(x)\big)$, for  types $\Sigma\subseteq \wp(A)$. Moreover, such a normal form is effective.
\end{corollary}

In our companion paper \cite{companionWEAK} on automata, we need versions of these results for 
the monotone and the continuous fragment.
For this purpose we define some slight modifications of the translations 
$(\cdot)^{\tmoda}$ and $(\cdot)^{\tmodb}$ which map positive and syntactically 
continuous sentences to respectively positive and syntactically continuous 
formulas.

\begin{theorem}
\label{t:inv1}
There are effective translations 
$(\cdot)^{\tmoda}: \ofoe^+ \to \ofo^{+} $
 and 
 $(\cdot)^{\tmodb}: {\ofoei}^+ \to \ofo^+$ 
such that 
$\phi \equiv \phi^{\tmoda}$ (respectively, $\phi \equiv \phi^{\tmodb}$) iff 
$\phi$ is invariant under quotients.
Moreover, we may assume that 
$(\cdot)^{\tmodb}: \cont{{\ofoei}(A)}{B}\cap {\ofoei}^+  \to \cont{\ofo(A)}{B} \cap \ofo^+$,
for any $B \subseteq A$.
\end{theorem}
\begin{proof}

We define translations $(\cdot)^{\tmoda}:\ofoe^+ \to\ofo^+$ and 
$(\cdot)^{\tmodb}: {\ofoei}^+ \to\ofo^+$ as follows. 
For $\ofoe^{+},{\ofoei}^{+}$-sentences in simple basic form we define
\[\begin{array}{lll}
     \Big( \posdbnfofoe{\vlist{T}}{\Pi} \Big)^{\tmoda} 
   & \isdef 
   & \bigwedge_{i} \exists x_i. \tau^+_{T_i}(x_i) 
     \land \forall x. \bigvee_{S\in\Pi} \tau^+_S(x),
\\   \Big( \posdbnfofoei{\vlist{T}}{\Pi}{\Sigma} \Big)^{\tmodb} 
   & \isdef 
   & \bigwedge_{i} \exists x_i. \tau^+_{T_i}(x_i) 
     \land \forall x. \bigvee_{S\in\Sigma} \tau^+_S(x),
\end{array}\]
and then we use, respectively, the Corollaries~\ref{cor:ofoepositivenf} and
\ref{cor:ofoeipositivenf} to extend these translations to the full positive
fragments $\ofoe^{+}$ and ${\ofoei}^{+}$, as we did in Definition~\ref{d:7-21}
for the full language.

We leave it as an exercise for the reader to prove the analogue of 
Proposition~\ref{p-1P} for these translations, and to show how the first
statements of the theorem follows from this.

Finally, to see why we may assume that $(\cdot)^{\tmodb}$ restricts to a map 
from the syntactically $B$-continuous fragment of ${\ofoei}^+(A)$ to the 
syntactically $B$-continuous fragment of ${\ofo}^+(A)$, assume that $\phi 
\in \ofoei(A)$ is continuous in $B \subseteq A$.
By Corollary~\ref{cor:ofoeicontinuousnf} we may assume that $\phi$ is a
disjunction of formulas of the form $\posdbnfofoei{\vlist{T}}{\Pi}{\Sigma}$,
where $B \cap \bigcup \Sigma = \nada$.
This implies that in the formula $\phi^{\tmodb}$ no predicate symbol $b \in
B$ occurs in the scope of a universal quantifier, and so $\phi^{\tmodb}$
is syntactically continuous in $B$ indeed.
\end{proof}


{\small
\bibliographystyle{plain}
\bibliography{logic}
}

\end{document}